\documentclass[journal, 10pt]{IEEEtran}

\IEEEoverridecommandlockouts
\usepackage{cite}
\usepackage{amsmath,amsthm,amsfonts,amssymb,bm,mathrsfs}
\usepackage{diagbox} 
\usepackage[procnumbered,linesnumbered,ruled,vlined]{algorithm2e}
\usepackage{etoolbox}

\makeatletter
\AtBeginEnvironment{procedure}{\let\c@algocf\c@procedure}
\makeatother
\SetProcNameSty{textsc}
\SetProcArgSty{textsc}
\usepackage{graphicx}
\usepackage{textcomp}
\usepackage{xcolor}
\usepackage{caption}
\usepackage{subfigure}
\usepackage{amsmath}
\usepackage{booktabs}
\graphicspath{{fig/}}
\allowdisplaybreaks[4]
\usepackage{geometry}
\def\BibTeX{{\rm B\kern-.05em{\sc i\kern-.025em b}\kern-.08em
    T\kern-.1667em\lower.7ex\hbox{E}\kern-.125emX}}

\usepackage{multirow}

\theoremstyle{remark}
\newtheorem{lemma}{\bfseries\textbf{Lemma}}
\newtheorem{definition}{\bfseries\textbf{Definition}}
\newtheorem{assumption}{\bfseries\textbf{Assumption}}

\newtheorem{theorem}{\bfseries\textbf{Theorem}}
\newtheorem{corollary}{\bfseries\textbf{Corollary}}[theorem]

\begin{document}

\title{RHFedMTL: Resource-Aware Hierarchical Federated Multi-Task Learning}

\author{Xingfu Yi, Rongpeng Li, Chenghui Peng, Fei Wang, Jianjun Wu, and Zhifeng Zhao

\thanks{X. Yi and R. Li are with Zhejiang University, Hangzhou, China
(email: \{yixingfu, lirongpeng\}@zju.edu.cn).}
\thanks{C. Peng, F. Wang, and J. Wu are with Huawei Technologies Co., Ltd.
(email: \{pengchenghui, wangfei76, wujianjun\}@huawei.com).}
\thanks{Z. Zhao is with Zhejiang Lab, Hangzhou, China
(email: zhaozf@zhejianglab.com).}
\thanks{Part of the paper has been accepted by IEEE PIMRC 2022 \cite{yi_hfedmtl_2022}.}
}
\maketitle

\begin{abstract}
    The rapid development of artificial intelligence (AI) over massive applications including Internet-of-things on cellular network raises the concern of technical challenges such as privacy, heterogeneity and resource efficiency.
    Federated learning is an effective way to enable AI over massive distributed nodes with security.
    However, conventional works mostly focus on learning a single global model for a unique task across the network, and are generally less competent to handle multi-task learning (MTL) scenarios with stragglers at the expense of acceptable computation and communication cost. Meanwhile, it is challenging to ensure the privacy while maintain a coupled multi-task learning across multiple base stations (BSs) and terminals. In this paper, inspired by the natural cloud-BS-terminal hierarchy of cellular works, we provide a viable resource-aware hierarchical federated MTL (RHFedMTL) solution to meet the heterogeneity of tasks, by solving different tasks within the BSs and aggregating the multi-task result in the cloud without compromising the privacy. Specifically, a primal-dual method has been leveraged to effectively transform the coupled MTL into some local optimization sub-problems within BSs. Furthermore, compared with existing methods to reduce resource cost by simply changing the aggregation frequency,
    we dive into the intricate relationship between resource consumption and learning accuracy, and develop a resource-aware learning strategy for local terminals and BSs to meet the resource budget.
    Extensive simulation results demonstrate the effectiveness and superiority of RHFedMTL in terms of improving the learning accuracy and boosting the convergence rate.
\end{abstract}
\begin{IEEEkeywords}
    Federated Learning, Multi-Task Learning, Mobile Edge Computing, Artificial intelligence (AI)
\end{IEEEkeywords}

\section{Introduction}
    Benefiting from the rapid development of cellular networks, artificial intelligence (AI) over massive Internet-of-things (IoT) or smart phones
    becomes possible, and there is coming to a consensus that 6G will revolutionize as a network of AI \cite{6Gvisions, intelligent5G, collective, management}.
    Meanwhile, network AI in 6G should simultaneously allow various sensing, mining, prediction and reasoning tasks across different industries and there emerges a common concern for storing and transmitting the rapidly expanded data through network in the near future~\cite{fogIot}.
    Besides, the heterogeneous and complex nature of Industrial IoT (IIoT) systems has brought together many technical challenges such as privacy, heterogeneity and resource efficiency~\cite{IIot}.
    
    In that regard, federated learning (FL) promises to enable AI over distributed terminals with privacy~\cite{fedavg} and has been extensively used in many scenarios~\cite{Keyboard,machine, wireless}.
    In particular, federated learning learns a shared model by aggregating local updates from terminals while keeping the training data on the terminals for the sack of privacy~\cite{fedavg}. Meanwhile, multi-task learning (MTL), which capably keeps more diversity with respect to models, can significantly improve the learning accuracy and the utilization of data~\cite{mtl,mocha}.
    However, conventional FL works mostly focus on learning a single unified model for one unique task across the network, and are generally less competent to handle MTL scenarios with stragglers at the expense of acceptable computation and communication cost.
    For example, \cite{adaptive} proposes to adjust the frequency of global aggregation so as to better leverage the available resource,
    but it only focuses on learning a single global model and cannot adjust the learning process of local terminals.
    \cite{energyefficient} proposes an energy-efficient computation and transmission resource allocation scheme over wireless communication networks while it does not support MTL and neglects the hierarchical network structure.
    \cite{TimeOptimization} proposes a probabilistic user selection scheme to reduce the FL convergence time and the FL training loss.
    \cite{convergence} proposes a resource allocation algorithm over
    wireless networks to capture the trade-off between the wall clock training time and the terminal energy consumption. Still, it focuses on a single-task scenario without considering the hierarchical structure. 
    As for the resource utilization optimization of FL in wireless networks, \cite{fedwireless} mainly derives an analytical model to characterize the performance of FL, while\cite{Differential} proposes a private gradient aggregation scheme for wireless federated learning.
    \cite{powercontrol} studies adaptive power allocation for distributed gradient descent.
    Though \cite{LearningAndCommunications} provides a comprehensive study of the connection between the performance of FL algorithms and the
    underlying wireless network, it only talks about the energy consumption of each user into consideration and cannot adjust the energy consumption of the complete learning process. 
    Foremost, within the scope of federated learning, although Stochastic Gradient Descent (SGD) has become popular for solving supervised machine learning optimization problems~\cite{fedavg,adaptive}, the coupling of multiple learning tasks makes it difficult to maintain the data privacy under the conventional FL framework.
    
    In this paper, we focus on a resource-aware hierarchical federated multi-task learning (RHFedMTL) framework, so as to perform federated MTL for large-scale data with limited resource budget and without compromising the privacy.
    Notably, Stochastic Dual Coordinate Ascent (SDCA) gives stronger convergence results than the primal-only methods (e.g., SGD), at the same iteration cost~\cite{sdca},
    and previous works of federated MTL such as CoCoA~\cite{cocoa,cocoa2} use the state-of-the-art primal-dual framework with SDCA. However, it only focuses on distributed optimization method and cannot handle the unique systems challenges (e.g., stragglers) of the federated environment~\cite{mocha}.
    MOCHA~\cite{mocha} attempts to present the federated multi-task learning and shows that MTL is a natural choice to handle statistical challenges in the federated setting.
    However, one task in MOCHA corresponds to one terminal, which becomes computation-intensive and even impractical, since the required computation resources grow exponentially along with the increase of the number of terminals.
    Meanwhile, in practice, it is quite common to allocate multiple terminals within one region (e.g., one BS) to perform one task together.
    Furthermore, both CoCoA and MOCHA neglect the technical challenges such as hierarchical resource management.
    Though~\cite{clientedge} introduces a client-edge-cloud hierarchical federated learning algorithm with multiple edge servers performing partial model aggregation to reduce communication cost, 
    it does not support MTL and thus is unable to cope with unique systems challenges (e.g., stragglers) of the FL environment.
    Besides, the resource management of previous works are static (i.e., not resource-aware) and cannot adjust the local learning process of terminals to balance the trade-off between resource cost and learning performance~\cite{strategies,efficient,adaptive,clientedge}.
    
    Therefore, inspired by the natural cloud-BS-terminal hierarchy of cellular works, we propose a resource-aware hierarchical federated MTL (RHFedMTL) method, where each BS is responsible for training one learning task with the aid of attached terminals while the cloud takes charge of aggregating the multi-task results. Furthermore, benefiting from primal-dual optimization method, 
    we transform the global primal optimization problem into separate dual sub-problems across BSs, so as to guarantee the privacy.
    Moreover, such a problem-splitting lays the very foundation for further developing a resource-aware learning strategy, as the proved relationship between the learning accuracy and the number of iterations as well as the trade-off between the terminal iteration and the BS iteration.
    Compared with the existing works, our main contributions in this paper are summarized as follows:
    \begin{enumerate}   
        \item
        We leverage a hierarchical structure for federated MTL on top of the primal-dual method SDCA, and propose a RHFedMTL method, which boosts the flexibility of MTL and enables the learning of different specified models over massive terminals without compromising the privacy.
        \item
        We provide the derived convergence bound of the vanilla HFedMTL, and unveil
        the relationship between the terminal iteration number and the BS iteration number required to converge.
        \item
        On top of the aforementioned derived relationship, within RHFedMTL, we propose an resource-aware algorithm RHFedMTL, which dynamically
        adapts the terminal iteration number and corresponding BS iteration number to balance the trade-off between learning accuracy and resource cost, so as
        to reach superior performance with limited resource budget.
        \item
        We evaluate the performance of the RHFedMTL algorithm via extensive simulations, and validate its effectiveness, robustness and resource-awareness. 
    \end{enumerate}

    The remainder of the paper is organized as follows.
    In Section II, we introduce the details of system model, 
    present some preliminaries of federated multi-task learning, and introduce the resource-constrained problem.
    In Section III, we propose the vanilla HFedMTL algorithm and derive the dual problem.
    Section IV gives the RHFedMTL algorithm and discusses how to dynamically set the terminal iteration number based on the provided convergence analysis.  
    Section V gives the simulation scenarios and demonstrates the numerical results.
    We conclude the paper in Section V.

\section{System Model and Problem Formulation}

\subsection{System Model}
We primarily consider a resource-intense hierarchical mobile computing environment with multiple task models to be learned.
In particular, we assume there exist $N$ BSs (i.e., $B_1, \cdots, B_N$) connected to the cloud server, each corresponding to one of $N$ coupled tasks.
Meanwhile, each BS $B_b$, $b \in \{ 1, \cdots, N\}$ covers $N_b$ terminals $T_{b,t}$, $t \in \{ 1, \cdots, N_b\} $ (e.g., smart phones, IoT devices), each collecting $S_{b,t}$ samples of data.
Targeting to learn models $\boldsymbol{w}_b$, $b \in \{1, \cdots, N\}$ for all $N$ tasks, the MTL can be formulated as computing model parameters from local datasets

\begin{figure}
    \centering
    \includegraphics[scale=0.33]{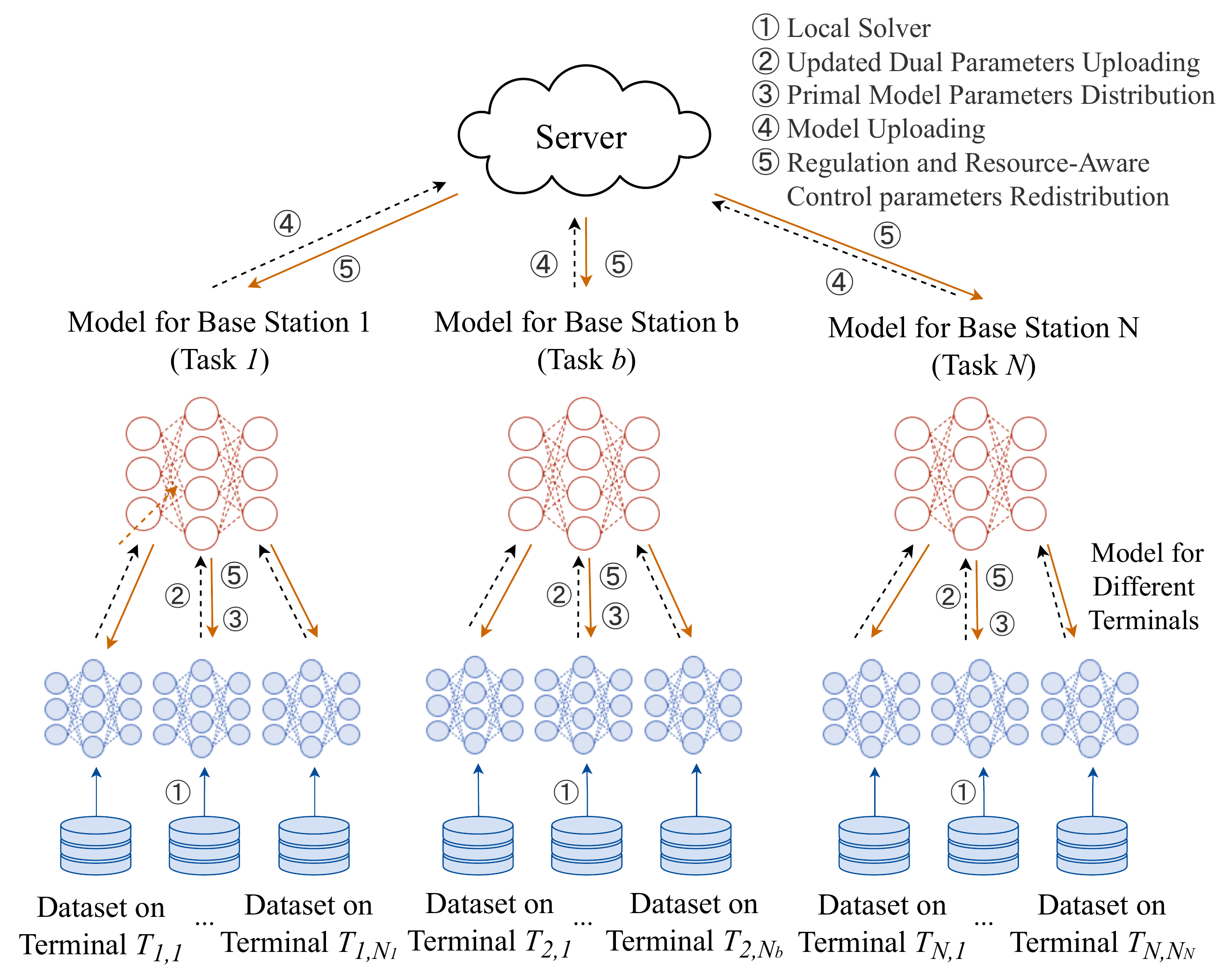}
    \caption{The system model of RHFedMTL.}
    \label{fg:system}
\end{figure}

\begin{table}
    \centering
    \caption{Notations.}
    \label{tb:notation}
    \begin{tabular}{lr}
        \toprule
        Notations                   & Description                                        \\
        \midrule
        $N$                         & Number of BSs                                      \\
        $N_b$                       & Number of Terminals under One BS                   \\
        $n_b$                       & Number of Data under One BS                        \\
        $n_{b,t}$                   & Dataset Size of Terminal $t$ under Task $b$        \\
        $\boldsymbol{w}_{b}$        & Model Parameters of Task $b$                       \\
        $\boldsymbol{\alpha}_{b}$   & Dual Parameters of All Data under Task $b$         \\
        $\boldsymbol{\alpha}_{[t]}$ & Dual Parameters of Local Data in Terminal $t$      \\
        $\mathcal{R}_b$             & MTL Regulation Function                            \\
        $\boldsymbol{r}_b$          & MTL Regulation Parameters                          \\
        $J$                         & Type of Considered Resources                       \\
        $C_{\text{bud}}$            & Resource Budget                                    \\
        $C_{\text{tol}}$            & Total Resource Consumption                         \\
        $C_{\text{dev}}$            & Standard Resource Cost of Single Device            \\
        $C_{\text{real}}$           & Resource Cost of Single Device                     \\
        $H$                         & No. of Terminal Iterations in a Specific Device \\
        \bottomrule
    \end{tabular}
\end{table}

\begin{align}
    \min_{\mathbf{W},\mathbf{R}}  \Bigg\{ & \frac{1}{N} \sum_{b=1}^N \bigg( \frac{1}{n_b} \sum_{i=1}^{n_b} \mathcal{L}( \boldsymbol{w}_{b}^\intercal \boldsymbol{x}_{b}^i)  + \frac{\lambda_1}{2}  \| \boldsymbol{w}_{b}\Vert^2 + \mathcal{R}_b(\boldsymbol{w}_{b}) \bigg)  \Bigg\}  \nonumber \\
    s.t. \quad                            & y_{b}^i = \boldsymbol{w}_{b}^\intercal \boldsymbol{x}_{b}^i \quad \forall i \in \{1,\ldots,n_b\}, b \in \{1,\ldots,N\}, \label{eq:mtl0}
\end{align}
where $n_b \triangleq \sum_{t=1}^{N_b} S_{b,t}$ represents the total number of data distributed under BS $b$ (i.e., task $b$) and the superscript $\intercal$ indicates the transpose operator. $y_{b}^i \triangleq \boldsymbol{w}_{b}^\intercal \boldsymbol{x}_{b}^i$ corresponds to the output for input of $\boldsymbol{x}_{b}^i$ in local dataset under BS $B_b$.
In summary, $\mathbf{W} \triangleq  [\boldsymbol{w}_{1} \cdots \boldsymbol{w}_{N_b}]$ and $\mathbf{R} \triangleq [\mathcal{R}_1, \ldots, \mathcal{R}_b, \ldots, \mathcal{R}_N]$ represents the weights of all BSs and the MTL regulation function for all tasks, respectively.
Consistent with the single-task learning, the MTL in \eqref{eq:mtl0} imposes an $\mathcal{L}_2$-regulation with constant $\lambda_1 > 0$. Nevertheless, MTL adds regulation functions $\mathcal{R}_b(\boldsymbol{w}_{b})$ to reflect the relationship among tasks.
Typically, assumptions on MTL regulation functions can be categorized into two groups, dependent on the a-priori existence of the relationships amongst tasks \cite{mocha}.
Specifically, \cite{multitask, multitask-feature} assume that the coupling structure between different tasks is known a priori
while \cite{multitask2} assumes unknown yet learnable relationships between different tasks.
Taking the example of the former case, consistent with \cite{zhou2010exclusive}, we could directly adopt an $\mathcal{L}_2$ regulation on $\boldsymbol{W}$ to constrain the difference of models amongst different tasks while reflecting the probably resemblance between the learning tasks (to a reference task).
Mathematically, such a regulation could be formulated as
\begin{equation}\label{eq:r}
    \mathcal{R}_b(\boldsymbol{w}_{b}) = \frac{\lambda_2}{2} \| \boldsymbol{w}_{b} - \boldsymbol{r}\Vert^2,
\end{equation}
where $\boldsymbol{r} \triangleq \frac{1}{N} \sum_{b=1}^N \boldsymbol{w_b}$ and $\lambda_2$ indicates the relative importance of the multi-task regulation to the overall loss in \eqref{eq:mtl0}.
Meanwhile, as for the second case, the multi-task regulation function $\mathcal{R}_b$ can be set as
\begin{equation}\label{eq:r2}
    \mathcal{R}_b(\boldsymbol{w}_{b}) = \frac{1}{N} \lambda_2 \text{tr} (\boldsymbol{W} \boldsymbol{\Omega}^{-1} \boldsymbol{W}^\intercal),
\end{equation}
where $\text{tr} (\cdot )$ denotes the trace of a matrix. Furthermore, the multi-task relationship matrix $\boldsymbol{\Omega}$ in \eqref{eq:r2} can be learned from data.
In this work, we primarily focus on the former case in \eqref{eq:r} while our work can be extended to the latter one with an additional learning loop.

\subsection{Problem Formulation}

In resource-intense mobile computing environment, it is of vital importance to limit the amount of resources to achieve a target error of the loss function, so as to keep the operational cost low on the basis of no system backlog. Hence, the formulation in \eqref{eq:mtl0} shall be complemented with the resource constraints.
Consistent with \cite{adaptive}, the terminology of ``resources" here is generic and can include time, energy \& economic cost related to both computation and communication.
Without loss of generality, assume that there exists $J$ types of resources and the resource cost for the same type of devices is equal.
$C_{j,\text{tol}}$ and $C_{j,\text{bud}}$, for $j \in \{1,\ldots, J\}$, represents the total resource consumption and the budget of type-$j$ resource respectively. We mainly focus on the resource consumption of terminals and BSs due to its natural importance in the mobile network.
Besides, the standard resource consumption of  type-$j$ resource  for a terminal and a BS to perform an iteration step is defined as $C_{j,\text{dev}}$ and $C_{j, \text{BS}}$ respectively. Furthermore, assume that an hierarchical iteration methodology is adopted here. In other words, one server iteration includes multiple BS iterations, while one BS iteration encompasses several terminal iteration. Consequently, on the basis of $M$ server iterations, for $K$ BS iterations and $H_b$ terminal iterations  under certain BS $B_b$ ($ b \in \{ 1, \ldots, N\} $), the type-$j \in \{1, \ldots, J \}$ resource consumption  shall be bounded as
\begin{align}
    K M \sum_{b=1}^N \big( C_{j,\text{BS}} + {N_b} H_b C_{j,\text{dev}} \big) \leq C_{j,\text{bud}}.
\end{align}
In other words, the MTL problem in \eqref{eq:mtl0} can be re-formulated as \eqref{eq:problem}. In the paper, for the coupled MTL, we resort to developing a resource-aware federated solution to yield a learning methodology with appropriate $K$ and $H_b$, on the basis of fundamental federated assumptions (i.e., no-violation of the privacy concerns).

\begin{figure*}[ht]
    \begin{align}
        \min_{\mathbf{W},\mathbf{R}} & \Bigg\{ \frac{1}{N} \sum_{b=1}^N \bigg( \frac{1}{n_b} \sum_{i=1}^{n_b} \mathcal{L}( \boldsymbol{w}_{b}^\intercal \boldsymbol{x}_{b}^i)  + \frac{\lambda_1}{2}  \| \boldsymbol{w}_{b}\Vert^2 + \mathcal{R}_b(\boldsymbol{w}_{b}) \bigg)  \Bigg\} \quad \forall K,H_b \in \{ 1, \ldots, \infty \} \notag \\
        s.t.                         & \quad  K M \sum_{b=1}^N \big( C_{j,\text{BS}} + {N_b} H_b C_{j,\text{dev}} \big) \leq C_{j,\text{bud}} \quad \forall j \in \{1, \ldots, J \}, \forall b \in \{ 1, \ldots, N\}  \notag                                                                                                          \\
                                     & \quad y_{b}^i = \boldsymbol{w}_{b}^\intercal \boldsymbol{x}_{b}^i \quad \forall i \in \{1,\ldots,n_b\}, b \in \{1,\ldots,N\} \label{eq:problem}
    \end{align}
    \hrulefill
\end{figure*}

\section{Hierarchical Federated Multi-Task Learning}
The coupled MTL models make it challenging to directly compute the parameters without knowing all the dataset distributed among terminals. In this part, inspired by the primal-dual methodology, we talk about how to formulate the dual formulation of \eqref{eq:problem} and decompose the global problem into distributed sub-problems without sacrificing the privacy of local terminals.

We define $\boldsymbol{\alpha}$ to be the concatenation of all the dual variables ${\alpha}_{b}^i$.
Mathematically, $\boldsymbol{\alpha} \triangleq [ \boldsymbol{\alpha}_1, \ldots, \boldsymbol{\alpha}_b, \ldots, \boldsymbol{\alpha}_N ]$,
where $\boldsymbol{\alpha}_b$ is the concatenation of all the dual variables under the same BS and can be represented in two inter-changeable forms:
(i) $\boldsymbol{\alpha}_b \triangleq [{\alpha}_b^1, \ldots, {\alpha}_b^{n_b}]$ is a simple combination of all the dual variables distributed among various terminals under one BS, where ${\alpha}_{b}^i$ is the dual variable for the data point $(\boldsymbol{x}_{b}^i,{y}_{b}^i)$.
(ii) $\boldsymbol{\alpha}_b \triangleq [\boldsymbol{\alpha}_{[1]}, \ldots, \boldsymbol{\alpha}_{[t]}, \ldots, \boldsymbol{\alpha}_{[N_b]} ]$ where $\boldsymbol{\alpha}_{[t]} \in \mathbb{R}^{S_{b,t}}$ is the concatenation of dual variables corresponding to the local dataset of terminal $T_{b,t}$.

We leverage the following lemma to demonstrate the equivalence between the primal problem and the aggregation of decomposed sub-problems at each BS, thus facilitating the subsequent analysis of an hierarchical federated MTL solution.

\begin{lemma}\label{lemmaDt}
    For any $\boldsymbol{\alpha}$, we have
    \begin{align}
        D(\boldsymbol{\alpha})
        = & \frac{1}{N} \sum_{b=1}^N  \bigg(  \frac{1}{n_b} \sum_{i=1}^{n_b} -\mathcal{L}^*(-{\alpha}_{b}^i) + \mathcal{R}_b^*(\boldsymbol{w}_{b})  \bigg) \label{eq:dual0} \\
        = & \frac{1}{N} \sum_{b=1}^N {D}_{b}(\boldsymbol{\alpha}_{b}),  \nonumber
    \end{align}
    where
    \begin{align}
           & {D}_{b}(\boldsymbol{\alpha}_{b})  \label{eq:localdual}                                                                                                                                                                                                                                 \\
        \triangleq & \frac{1}{n_b} \sum_{i=1}^{n_b} -\mathcal{L}^*( - {\alpha}_{b}^i )  + \frac{\lambda_2^2 \left\lVert \boldsymbol{r}_b \right\rVert^2 - \left\lVert \boldsymbol{A}_b \boldsymbol{\alpha}_b \right\rVert^2 -  2 \lambda_2 \boldsymbol{A}_b \boldsymbol{\alpha}_b \boldsymbol{r}_b^\intercal    }{2(\lambda_1 + \lambda_2)} \notag
    \end{align}
    and $ \boldsymbol{A}_b \in \mathbb{R}^{d \times n_b} $ collects data examples $ A_i = \frac{1}{n_b} \boldsymbol{x}_b^i$ in its columns.
    Given dual variables $\boldsymbol{\alpha}_b$, corresponding primal variables can be found via $\boldsymbol{w}_{b}  = \frac{1}{\lambda_1 + \lambda_2} \big( \lambda_2 \boldsymbol{r}_b + \boldsymbol{A}_b \boldsymbol{\alpha}_b \big)$. Moreover, $\mathcal{L}^*$ and $\mathcal{R}_b^*$ are the conjugate dual functions of $\mathcal{L}$ and $\mathcal{R}_b$ respectively.
\end{lemma}
\begin{proof}
    The proof is a direct application of the definition of the conjugate function and we leave it in Appendix.
\end{proof}

Lemma \ref{lemmaDt} implies the MTL problem can be solved according to several sub-problems by individual BSs. Next, we resort to the following methods to tackle each sub-problem.
For each sub-problem of \eqref{eq:localdual}, \textit{coordinate ascent} methods have proven to be very efficient since it requires no step size, and has a well-defined stopping criterion given by the duality gap \cite{cocoa}.
We use \textit{stochastic dual coordinate descent} (SDCA) method which has proven very suitable for use in large-scale problems, while giving stronger convergence results than the primal-only methods (e.g., SGD) at the same iteration cost \cite{sdca}.
Given that the formulation of
\eqref{eq:localdual} allows terminals to find updates $\Delta \alpha_{b}^i$ to the dual variables in $\boldsymbol{\alpha}_b$ by only accessing the data $(\boldsymbol{x}_{b}^i,{y}_{b}^i)$ stored locally.
In other words, the terminals under the same BS can solve the related part of the subproblem locally and independently, and the algorithm could converge at a higher speed compared to gradient decent methods.
Therefore, with selected data $(\boldsymbol{x}_{b,t}^i,y_{b,t}^i)$ in the training process, by Lemma \ref{lemmaDt}, terminal-oriented sub-problem can be formulated as
\begin{equation}\label{eq:local}
    \begin{split}
        \max_{\Delta \alpha_{b}^i} \bigg\{&  - \mathcal{L}^*( - ( {\alpha}_{b}^i + \Delta \alpha_{b}^i) )  + \frac{n_b}{{2(\lambda_1 + \lambda_2)}} \bigg[ \lambda_2^2 \left\lVert \boldsymbol{r}_b \right\rVert^2  \\
            & - \left\lVert \boldsymbol{w}_b + \frac{1}{(\lambda_1 + \lambda_2) n_b} \Delta \alpha_b^i \boldsymbol{x}_{b,t}^i \right\rVert^2  \\
            & - 2 \lambda_2 \Big( \boldsymbol{w}_b + \frac{1}{(\lambda_1 + \lambda_2) n_b} \Delta \alpha_b^i \boldsymbol{x}_{b,t}^i \Big) \boldsymbol{r}_b^\intercal  \bigg]  \bigg\}.
    \end{split}
\end{equation}

Finally, the MTL problem in \eqref{eq:dual0} has been split into several sub-problems across distributed BSs, each of which can be tackled by terminals locally. Correspondingly, a hierarchical federated MTL (HFedMTL) solution, which involves a hierarchical iteration mechanism including the server iteration, the BS iteration and the terminal iteration, can be attained.
\begin{itemize}
    \item As for each terminal iteration of the training process, terminals conduct their local training process to maximize the local optimizing function \eqref{eq:local} and communicate with their BS periodically.
    \item During each BS iteration, the BSs update and distribute the model parameters after aggregating the uploaded intermediate variables of all terminals. Subsequently, the BSs upload their model to the server at a lower frequency. 
    \item Based on the received models, the server updates the parameters of regulation function $\mathcal{R}_b$ and sends the updated parameters back to terminals through their connected BSs. During the server iteration, the server generates and sends the regulation parameters to all terminals (via BSs) using the uploaded model parameters of all tasks.
\end{itemize}
In a summary, for each server iteration, the BSs will perform $K$ BS iterations; while for each iteration of BS $B_b$, $ b \in \{ 1, \ldots, N\} $, the terminals will perform $H_b$ terminal iterations.

\section{Resource-Aware HFedMTL Method}\label{RHFedMTL}
\subsection{Convergence Analysis of HFedMTL}
Before delving into the details of resource-aware HFedMTL, we first present the convergence analysis of HFedMTL, which will guide the designs of a resource aware solution.

Beforehand, we give a definition of the duality gap after one terminal iteration.
\begin{definition}
    Within each BS iteration, we define the duality gap of sub-problems specifying how far we are from the optimum on terminal $t$ with all other terminals fixed.
    Mathematically, given $\boldsymbol{\alpha}_{[1]} \ldots \boldsymbol{\alpha}_{[N_b]}$ from $N_b$ terminals,
    \begin{equation} \label{eq:eps}
        \begin{split}
            \mathcal{E}_{{D}_{b,t}}{ {(\boldsymbol{\alpha}_b)}} \triangleq \max_{\hat{\boldsymbol{\alpha}}_{[t]}} {D}_b( \boldsymbol{\alpha}_{[1]}, \ldots, \hat{\boldsymbol{\alpha}}_{[t]}, \ldots, \boldsymbol{\alpha}_{[N_b]}) - &\\
            {D}_b( \boldsymbol{\alpha}_{[1]}, \ldots, {\boldsymbol{\alpha}}_{[t]}, \ldots, \boldsymbol{\alpha}_{[N_b]})&
        \end{split}
    \end{equation}
\end{definition}
We also have the following assumption of the updated duality gap.
\begin{assumption}\label{assum1}
    We assume that there exists $\Theta_{b} \in (0,1]$ such that for any given $\boldsymbol{\alpha}_b$, the subproblem running on terminal $t$ alone returns a (possibly random) update $\Delta \boldsymbol{\alpha}_{[t]}$ after $H_b$ terminal iterations such that the expectation of updated duality gap is bounded as:
    \begin{equation}
        \mathbf{E} \Big[ \mathcal{E}_{{D}_{b,t}}( \boldsymbol{\alpha}_{[1]}, \ldots, {\boldsymbol{\alpha}}_{[t]} + \Delta \boldsymbol{\alpha}_{[t]}, \ldots, \boldsymbol{\alpha}_{[N_b]}) \Big] \leq \Theta_{b} \cdot \mathcal{E}_{{D}_{b,t}} (\boldsymbol{\alpha}_b)
    \end{equation}
    Assume the loss functions $\mathcal{L}$ is ($1/\gamma$)-smooth.
    Then, consistent with  \cite[Proposition 1]{cocoa},
    \begin{equation} \label{eq:Theta}
        \Theta_{b} = \left(1 - \frac{(\lambda_1 + \lambda_2) n_b \gamma}{1 + (\lambda_1 + \lambda_2) n_b \gamma}  \frac{1}{\tilde{n}_b}  \right)^{H_b}
    \end{equation}
    where $\tilde{n}_b \triangleq \max_t{n_{b,t}}$ is the size of the largest local dataset among terminals.
    Specifically, $\Theta_{b} \triangleq 1$ means that the terminals under task $b$ made no updates and therefore no resource consumption,
    while $\Theta_{b} \rightarrow 0$ implies that the duality gap comes to zero, which is unrealistic as the number of terminal iterations (and correspondingly the resource consumption) grows to infinity in this case.
    Therefore, $\Theta_{b}$ indicates the resource required by the system.
    We also define $\Theta \triangleq \max_b{\Theta_b}$ to facilitate subsequent operations.
\end{assumption}

The following theorem gives the convergence analysis of HFedMTL with respect to the number of BS iteration $K$.
\begin{theorem}\label{the}
    Assume the loss functions $\mathcal{L}$ is $(1/\gamma)$-smooth, choosing BS iteration number $K$ such that
    \begin{equation}\label{eq:k}
        \begin{split}
            K > \Big(1 - (1 - \Theta) \frac{\eta^*}{T^*} \Big) \log{\frac{\sum_{b=1}^N n_b}{N \epsilon_D}},
        \end{split}
    \end{equation}
    we have
    \begin{equation}\label{eq:duality_gap}
        \mathbf{E} \big[ D({\boldsymbol{\alpha}}^{*}) - D({\boldsymbol{\alpha}}^{(K)}) \big] \leq \epsilon_D,
    \end{equation}
    where $T^* \triangleq \max_{b} N_b$, $\eta^* \triangleq \min_{b} \eta_b$ and $\eta_b = \frac{(\lambda_1 + \lambda_2)  \gamma}{n_b \sigma  + (\lambda_1 + \lambda_2 ) \gamma}$ with $\sigma \geq \max_b{{\sigma_b}_{\min}} $ and
    \begin{equation}\label{eq:sigma}
        \begin{split}
            {\sigma_b}_{\min} \triangleq \max_{\alpha}  n_b^2
            \frac{ \sum_{t=1}^{N_b} \left\lVert \boldsymbol{A}_{[t]}  \boldsymbol{\alpha}_{[t]} \right\rVert ^2 - \left\lVert \boldsymbol{A}_b \boldsymbol{\alpha}_b \right\rVert ^2 }{ \left\lVert \boldsymbol{\alpha}_b \right\rVert ^2 }
        \end{split}
    \end{equation}

\end{theorem}
\begin{proof}

    According to Lemma \ref{lemmaDt}, the value of dual problem after $(k+1)$-th BS iteration can be expressed as

    \begin{equation}
        \begin{split}
            {D} \big(\boldsymbol{\alpha}^{(k+1)} \big)
            =& \frac{1}{N} \sum_{b=1}^N {D}_{b} \Big( \boldsymbol{\alpha}_{b}^{(k)} + \frac{1}{N_b} \sum_{t=1}^{N_b}  \Delta \boldsymbol{\alpha}_{[t]}  \Big)\\
        \end{split}
    \end{equation}
    Using the concavity of dual function $D$,
    \begin{equation}
        {D} \big(\boldsymbol{\alpha}^{(k+1)} \big) \geq  \frac{1}{N}  \sum_{b=1}^N \frac{1}{N_b} \sum_{t=1}^{N_b} {D}_{b} \Big( \boldsymbol{\alpha}_{b}^{(k)} +  \Delta \boldsymbol{\alpha}_{[t]}  \Big)
    \end{equation}

    Denoting $\hat{\alpha}_{[t]}^{*}$ to be the local maximizer as in \eqref{eq:eps}, we have
    \begin{align}
             & \mathbf{E} \Big[ {D} \big( \boldsymbol{\alpha}^{(k+1)} \big) - {D} \big(\boldsymbol{\alpha}^{(k)} \big)  \Big]\notag                                                                                                                             \\
        \geq & \frac{1}{N} \sum\nolimits_{b=1}^N \frac{1}{N_b} \sum\nolimits_{t=1}^{N_b} \Big\{ {D}_{b} \big( \boldsymbol{\alpha}_{b}^{(k)} + \Delta \boldsymbol{\alpha}_{[t]}  \big) - {D}_{b} \big(\boldsymbol{\alpha}_{b}^{(k)} \big)  \Big\} \notag         \\
        \geq & \frac{1}{NN_b} \sum\nolimits_{b=1}^N \sum\nolimits_{t=1}^{N_b} \Big\{   {D}_{b} \big( \boldsymbol{\alpha}_{b}^{(k)} + \Delta \boldsymbol{\alpha}_{[t]}  \big) \notag                                                                             \\
             & -  {D}_b( \boldsymbol{\alpha}_{[1]}^{(k)}, \ldots, \hat{\boldsymbol{\alpha}}^*_{[t]}, \ldots, \boldsymbol{\alpha}_{[N_b]}^{(k)})   \Big. \notag                                                                                                  \\
             & \Big.  + {D}_b( \boldsymbol{\alpha}_{[1]}^{(k)}, \ldots, \hat{\boldsymbol{\alpha}}^*_{[t]}, \ldots, \boldsymbol{\alpha}_{[N_b]}^{(k)})  -  {D}_{b} \big(\boldsymbol{\alpha}_{b}^{(k)} \big)  \Big\}                                              \\
        \geq & \frac{1}{NN_b} \sum\nolimits_{b=1}^N \sum\nolimits_{t=1}^{N_b} \Big\{  \mathcal{E}_{{D}_{b,t}}{ {(\boldsymbol{\alpha}_b^{(k)})}}  - \mathcal{E}_{{D}_{b,t}}{ {(\boldsymbol{\alpha}_b^{(k)} + \Delta \boldsymbol{\alpha}_{[t]}  )}} \Big\}. \notag
    \end{align}
    Meanwhile, under Assumption 1,
    \addtolength{\rightmargin}{0.05in}
    \begin{align}
             & \mathbf{E} \Big[ {D} \big( \boldsymbol{\alpha}^{(k+1)} \big) - {D} \big(\boldsymbol{\alpha}^{(k)} \big) \big| \boldsymbol{\alpha}^{(k)} \Big] \notag       \\
        \geq & \frac{(1 - \Theta)}{N} \sum\nolimits_{b=1}^N \frac{1}{N_b} \sum\nolimits_{t=1}^{N_b} \mathcal{E}_{{D}_{b,t}}{ {(\boldsymbol{\alpha}_b^{(k)})}}. \label{123}
    \end{align}

    Recalling the definition of $D_b(\boldsymbol{\alpha}_b)$ in \eqref{eq:localdual}, we can obtain \eqref{eq:e}, where the equalities $(a)$ to $(c)$ come from Lemma \ref{lemmaDt}. Besides, the inequality $(d)$ is due to \eqref{eq:theorem2_cocoa} in \cite[Theorem 2]{cocoa} with $\sigma$ in \eqref{eq:sigma}, and the inequality $(e)$ is given by introducing an extra $\eta_b \in [0,1]$ to link $\boldsymbol{\alpha}_b^{*}$ (i.e., the maximizer of \eqref{eq:e}). After applying the property of $(1/\gamma)$-smooth function and simple mathematical manipulations, we have the inequalities $(f)$ and $(g)$, respectively.

    \begin{figure*}
        \begin{align}
                                 & \sum_{t=1}^{N_b} \mathcal{E}_{{D}_{b,t}}{ {(\boldsymbol{\alpha}_b^{(k)})}} \notag                                                                                                                                                                                                                                                                                                                                                                           \\
            \stackrel{(a)}{=}    & \max_{\hat{\boldsymbol{\alpha}}_b}  \bigg\{ \sum\nolimits_{t=1}^{N_b} \Big[  {D}_b( \boldsymbol{\alpha}_{[1]}^{(k)} , \ldots, \hat{\boldsymbol{\alpha}}_{[t]}, \ldots, \boldsymbol{\alpha}_{[N_b]}^{(k)} ) -  {D}_b( \boldsymbol{\alpha}_{[1]}^{(k)}, \ldots, {\boldsymbol{\alpha}}_{[t]}^{(k)}, \ldots, \boldsymbol{\alpha}_{[N_b]}^{(k)}) \Big] \bigg\} \notag                                                                                            \\
            \stackrel{(b)}{=}    & \max_{\hat{\boldsymbol{\alpha}}_b} \Bigg\{  \frac{1}{n_b} \sum\nolimits_{i=1}^{n_b} \Big( -\mathcal{L}^*( - \hat{\alpha}_{b}^i )  + \mathcal{L}^*( - {\alpha}_{b}^{i^{(k)}} ) \Big) +  \frac{1}{2(\lambda_1 + \lambda_2)}  \sum\nolimits_{t=1}^{N_b}   \bigg( - 2 \lambda_2 \boldsymbol{A}_{[t]} \big( \hat{\boldsymbol{\alpha}}_{[t]} - {\boldsymbol{\alpha}}_{[t]}^{(k)}  \big) \boldsymbol{r}_b^\intercal \notag                                                 \\
                                 & \qquad- \left\lVert \boldsymbol{A}_b \boldsymbol{\alpha}_{b}^{(k)} +  \boldsymbol{A}_{[t]} (\hat{\boldsymbol{\alpha}}_{[t]} - \boldsymbol{\alpha}_{[t]}^{(k)} )  \right\rVert^2  + \left\lVert \boldsymbol{A}_b \boldsymbol{\alpha}_{b}^{(k)} \right\rVert^2  \bigg) \Bigg\}           \notag                                                                                                                                                               \\
            \stackrel{(c)}{=}    & \max_{\hat{\boldsymbol{\alpha}}_b} \Bigg\{ D_b(\hat{\boldsymbol{\alpha}_b}) - D_b(\boldsymbol{\alpha}^{(k)}_b)  + \frac{\left\lVert \boldsymbol{A}_b \hat{\boldsymbol{\alpha}}_{b}  \right\rVert^2 -  \lVert \boldsymbol{A}_b {\boldsymbol{\alpha}}_{b}^{(k)}  \rVert^2 }{2(\lambda_1 + \lambda_2)}  \notag                                                                                                                                                 \\
                                 & \qquad  + \frac{1}{2(\lambda_1 + \lambda_2)} \sum_{t=1}^{N_b} \bigg[  - \left\lVert \boldsymbol{A}_b \boldsymbol{\alpha}_{b}^{(k)} +  \boldsymbol{A}_{[t]} (\hat{\boldsymbol{\alpha}}_{[t]} - \boldsymbol{\alpha}_{[t]}^{(k)} )  \right\rVert^2 + \left\lVert \boldsymbol{A}_b \boldsymbol{\alpha}_{b}^{(k)} \right\rVert^2   \bigg]   \Bigg\} \label{eq:e}                                                                                                 \\
            \stackrel{(d)}{\geq} & \max_{\hat{\boldsymbol{\alpha}}_b} \bigg\{ D_b(\hat{\boldsymbol{\alpha}_b}) - D_b(\boldsymbol{\alpha}^{(k)}_b)  - \frac{\sigma }{2(\lambda_1 + \lambda_2)n_b^2} \left\lVert \hat{\boldsymbol{\alpha}}_{b} - \boldsymbol{\alpha}_{b}^{(k)} \right\rVert^2 \bigg\} \notag                                                                                                                                                                                     \\
            \stackrel{(e)}{\geq} & \max_{\eta_b \in [0,1]} \bigg\{ D_b( \eta_b \boldsymbol{\alpha}_b^{*} + (1 - \eta_b) \boldsymbol{\alpha}_b^{(k)} ) - D_b(\boldsymbol{\alpha}_b^{(k)})   - \frac{\sigma }{2(\lambda_1 + \lambda_2)n_b^2} \left\lVert \eta_b \boldsymbol{\alpha}_b^{*} + (1 - \eta_b) \boldsymbol{\alpha}_b^{(k)} - \boldsymbol{\alpha}_{b}^{(k)} \right\rVert^2  \bigg\} \notag                                                                                              \\
            \stackrel{(f)}{\geq} & \max_{\eta_b \in [0,1]}  \bigg\{ \eta_b D_b(\boldsymbol{\alpha}_b^*) + (1 - \eta_b) D_b(\boldsymbol{\alpha}_b^{(k)}) - D_b(\boldsymbol{\alpha}_b^{(k)}) +  \frac{\gamma \eta_b (1 - \eta_b)}{2 n_b} \left\lVert {\boldsymbol{\alpha}}_{b}^* - \boldsymbol{\alpha}_{b}^{(k)} \right\rVert^2   - \frac{\eta_b^2 \sigma }{2(\lambda_1 + \lambda_2)n_b^2} \left\lVert {\boldsymbol{\alpha}}_{b}^* - \boldsymbol{\alpha}_{b}^{(k)} \right\rVert^2 \bigg\} \notag \\
            \stackrel{(g)}{\geq} & \max_{\eta_b \in [0,1]} \bigg\{ \eta_b \big( D_b(  \boldsymbol{\alpha}_b^{*})  - D_b(\boldsymbol{\alpha}_b^{(k)}) \big)  + \frac{ \eta_b}{2 n_b}  \Big( \gamma (1 - \eta_b) - \frac{\eta_b \sigma }{ (\lambda_1 + \lambda_2) n_b } \Big)  \left\lVert {\boldsymbol{\alpha}}_{b}^* - \boldsymbol{\alpha}_{b}^{(k)} \right\rVert^2  \bigg\} \notag
        \end{align}
        \hrulefill
    \end{figure*}

    \begin{figure*}
        \begin{align}
            \max_{\hat{\boldsymbol{\alpha}}_b} \bigg\{ \left\lVert \boldsymbol{A}_b \hat{\boldsymbol{\alpha}}_{b}  \right\rVert^2 -  \lVert \boldsymbol{A}_b {\boldsymbol{\alpha}}_{b}^{(k)}  \rVert^2  + \sum_{t=1}^{N_b} \bigg[   - \left\lVert \boldsymbol{A}_b \boldsymbol{\alpha}_{b}^{(k)} +  \boldsymbol{A}_{[t]} (\hat{\boldsymbol{\alpha}}_{[t]} - \boldsymbol{\alpha}_{[t]}^{(k)} )  \right\rVert^2  + \left\lVert \boldsymbol{A}_b \boldsymbol{\alpha}_{b}^{(k)} \right\rVert^2  \bigg] \bigg\}  \geq \max_{\hat{\boldsymbol{\alpha}}_b} \Big\{  -   \frac{\sigma}{n_b^2} \left\lVert \hat{\boldsymbol{\alpha}}_{b} - \boldsymbol{\alpha}_{b}^{(k)} \right\rVert^2   \Big\} \label{eq:theorem2_cocoa}
        \end{align}
        \hrulefill
    \end{figure*}

    Letting $ (\gamma (1 - \eta_b) - \frac{\eta_b \sigma }{ (\lambda_1 + \lambda_2) n_b } )= 0 $, we have $\eta_b = \frac{(\lambda_1 + \lambda_2) n_b \gamma}{n_b \sigma  + (\lambda_1 + \lambda_2 ) n_b \gamma}$.
    Thus
    \begin{equation}\label{eq:efinal}
        \begin{split}
            \sum\nolimits_{t=1}^{N_b} \mathcal{E}_{{D}_{b,t}}{ {(\boldsymbol{\alpha}_b^{(k)})}}  \geq & \eta_b \Big( D_b(\boldsymbol{\alpha}_b^*) - D_b(\boldsymbol{\alpha}_b^{(k)}) \Big).\\
        \end{split}
    \end{equation}
    Substituting \eqref{eq:efinal} into \eqref{123}, we can derive
    \begin{equation}
        \begin{split}
            &\mathbf{E}  \Big[ D(\boldsymbol{\alpha}^{(k+1)}) - D(\boldsymbol{\alpha}^{(k)}) | \boldsymbol{\alpha}^{(k)} \Big] \\
            \geq &  \frac{(1 - \Theta)}{N} \sum\nolimits_{b=1}^N \frac{1}{N_b} \eta_b \Big( D_b(\boldsymbol{\alpha}^*) - D_b(\boldsymbol{\alpha}^{(k)}) \Big) \\
            \geq & \beta \Big(  D(\boldsymbol{\alpha}^*) - D(\boldsymbol{\alpha}^{(k)})  \Big),
        \end{split}
    \end{equation}
    where $\beta \triangleq (1 - \Theta) \frac{\eta^*}{T^*}$, $T^* \triangleq \max_b{N_b}  $ and $\eta^* \triangleq \min_b{\eta_b}$. Therefore,
    \begin{equation}
        \begin{split}
            &\mathbf{E} \Big[ D(\boldsymbol{\alpha}^*) - D(\boldsymbol{\alpha}^{(k+1)}) \Big] \\
            =& \mathbf{E} \bigg[ D(\boldsymbol{\alpha}^*) - D(\boldsymbol{\alpha}^{(k)}) - \Big( D(\boldsymbol{\alpha}^{(k+1)}) - D(\boldsymbol{\alpha}^{(k)}) \Big) \bigg]  \\
            \leq & (1 - \beta) \Big(D(\boldsymbol{\alpha}^*) - D(\boldsymbol{\alpha}^{(k)}) \Big).
        \end{split}
    \end{equation}
    Thus, we have
    \begin{equation}
        \begin{split}
            \mathbf{E} \Big[ D(\boldsymbol{\alpha}^*) - D(\boldsymbol{\alpha}^{(k)}) \Big] \leq  ( 1 - \beta )^k \Big(D(\boldsymbol{\alpha}^*) - D(\boldsymbol{\alpha}^{(0)}) \Big). \notag
        \end{split}
    \end{equation}
    Consistent with \cite{mocha}, we use the bound on the initial duality gap proved in \cite[Lemma 10]{cocoa2}, which states that $ {D}_{b}(\boldsymbol{\alpha}_{b}^{*}) - {D}_{b}(\boldsymbol{\alpha}_{b}^{(0)}) \leq n_b$.
    Given target convergence duality gap $\epsilon_D$, we can derive that
    \begin{equation}
        \begin{split}
            \epsilon_D \leq {(1 - \beta)}^K \frac{1}{N} \sum\nolimits_{b=1}^{N} n_b.
        \end{split}
    \end{equation}
    Finally, since when $x > 0$, $x > \log{x}$, we have
    \begin{equation}\label{eq:last}
        K \geq \frac{\log{\frac{N \epsilon_D}{\sum_{b=1}^N n_b}}}{\log{(1 - \beta)}}
        = \frac{\log{\frac{\sum_{b=1}^N n_b}{N \epsilon_D}}}{\log{\frac{1}{1 - \beta}}}
        > (1 - \beta) \log{\frac{\sum_{b=1}^N n_b}{N \epsilon_D}}. \notag
    \end{equation}
\end{proof}

\noindent\textit{Remark}: Theorem \ref{the} shows that the value of $\beta$ is mostly affected by the hardest task and the learning problem will always converge after an update of $\mathcal{R}_b$ if we set the number for BS iterations $K$ sufficiently large.
Hence, HFedMTL is robust to the change of multi-task learning parameters and guaranteed to converge.

Together with \eqref{eq:Theta}, Theorem \ref{the} implies the following corollary.
\begin{corollary}
    The minimal $K$ increases given a decrease of terminal iteration number $H_b$, so there exists a trade-off between the terminal iteration number $H_b$ and BS iteration number $K$.
\end{corollary}

\begin{figure*}
    \begin{align} \label{eq:rH}
        M \sum_{b=1}^N \big( C_{j,\text{BS}} + {N_b} H_b C_{j,\text{dev}} \big) \bigg\{ 1 -  \frac{\eta^*}{T^*} \Big[ 1 - \big(1 - \frac{(\lambda_1 + \lambda_2) n_b \gamma}{1 + (\lambda_1 + \lambda_2) n_b \gamma}  \frac{1}{\tilde{n}_b}  \big)^{\min_b{H_b}} \Big]  \bigg\} \log{\frac{\sum_{b=1}^N n_b}{N \epsilon_D}} < C_{j,\text{bud}}
    \end{align}
    \hrulefill
\end{figure*}

\subsection{The Description of RHFedMTL}
To meet all the requirements of the resource-aware problem defined in \eqref{eq:problem},
we extend the vanilla HFedMTL method to a resource-aware one so as to optimize the system model at a minimum resource cost.

Different form HFedMTL pre-defining all parameters, the resource-aware HFedMTL method is capable to be aware of the costs and budget of the system resources,
so as to dynamically adjust the terminal iteration number $H_b$ under BS $B_b$, $ b \in \{1, \ldots, N\}$ and BS iteration number $K$ while meeting the convergence target of duality gap in \eqref{eq:duality_gap}.
Recalling Theorem \ref{the} that a larger terminal iteration number $H_b$ increases the resource consumption linearly, but decreases the incurred duality gap non-linearly,
given the resource budget, there must exists a range of feasible $H_b$, within which choosing a bigger $H_b$ will increase the terminal iteration cost, but the decreased BS iteration number $K$ required for convergence will decrease the overall resource consumption.
The following theorem verifies the aforementioned intuitions.

\begin{algorithm}[t] \SetKwInOut{Initialize}{Initialize} \SetKwData{Left}{left}\SetKwData{This}{this}\SetKwData{Up}{up} \SetKwFunction{Union}{Union}\SetKwFunction{FindCompress}{FindCompress} \SetKwInOut{Input}{input}\SetKwInOut{Output}{output}
    \caption{Resource-aware HFedMTL method}
    \label{al:rhfedmtl}
    \KwData{$(\boldsymbol{x}_{b,t}^i,y_{b,t}^i)$, $ t \in (1,2,\ldots,N_b)$, $b \in (1, \ldots, N)$. Each terminal $T_{b,t}$ contains a local dataset containing $S_{b,t}$ samples of data}
    \KwIn{$C_{j,\text{dev}}$, $C_{j,\text{BS}}$, $C_{j,\text{bud}}$}
    \Initialize{$\boldsymbol{\alpha}_{b}^{(0)} \triangleq 0 $, $\boldsymbol{w}_{b} \triangleq 0$, $\boldsymbol{r}_b \triangleq 0$, $\forall b \in (1, \ldots, N)$}
    \BlankLine
    \For(server iteration){$j = 0,1,\ldots$}{
    server send regulation parameters $\boldsymbol{r}_b$ to terminals through their connected BSs after gathering the uploaded information\\
    \For{BS (i.e., tasks) $B_b$, $b \in (0,1,\ldots,N)$ in parallel}{
    send $T^*$, $\eta^*$ to all terminals \\
    \For(BS iteration){$k = 0,1,\ldots, K$}{
    \For{all terminals $T_{b,t}$, $t \in \{ 1 \cdots N_b\} $ in parallel}{
    ($\Delta \boldsymbol{\alpha}_{[t]}$, $\Delta \boldsymbol{w}_{b,t}$) $\leftarrow$ \\
    \textsc{ResourceSavingMethod} ($\boldsymbol{\alpha}_{[t]}$, $\boldsymbol{w}_b$, $\boldsymbol{r}_b$, $C_{j,\text{dev}}$, $C_{j,\text{BS}}$, $C_{j,\text{bud}}$, $T^*$, $\eta^*$)

    }
    $\boldsymbol{w}_{b} \leftarrow  \frac{1}{\lambda_1 + \lambda_2} ( \lambda_2  \boldsymbol{r}_b + \boldsymbol{w}_{b} + \frac{1}{N_b}  \sum_{t=1}^{N_b}  \Delta \boldsymbol{w}_{b,t} )$\\
    }
    }
    update regulation parameters $\boldsymbol{r}_b \leftarrow \frac{1}{N} \sum_{b=1}^N \boldsymbol{w_b}$\\
    update $T^* \leftarrow \max_{b} N_b$, $\eta^* \leftarrow \min_{b} \frac{(\lambda_1 + \lambda_2)  \gamma}{n_b \sigma  + (\lambda_1 + \lambda_2 ) \gamma}$.
    }
    \KwOut{$\boldsymbol{w}_{b}$, $\boldsymbol{r}_b$, $\forall b \in (1, \ldots, N)$}
\end{algorithm}

\DecMargin{1em}

\begin{procedure}[!htp] \SetKwInOut{Initialize}{Initialize}
    \caption{LocalDualMethod ()}\label{localresource}
    \KwData{Local data $ \{ (\boldsymbol{x}_{b,t}^i,y_{b,t}^i) \}_{i=1}^{S_{b,t}} $}
    \KwIn{$\boldsymbol{\alpha}_{[t]}$, $\boldsymbol{w}_b$, $\boldsymbol{r}_b$, $C_{j,\text{dev}}$, $C_{j,\text{BS}}$, $C_{j,\text{bud}}$, $T^*$, $\eta^*$}
    \Initialize{
    $\Delta \boldsymbol{\alpha}_{[t]} \leftarrow 0$, $\Delta \boldsymbol{w}_{b,t} \leftarrow 0$, $\boldsymbol{w}_{b,t} \leftarrow \boldsymbol{w}_{b}$,
    $f(H) \triangleq (  {N_b} H C_{j,\text{dev}} + N C_{j,\text{BS}} ) ( 1 -  \frac{\eta^*}{T^*} \big[ 1 - (1 - \frac{(\lambda_1 + \lambda_2) n_b \gamma}{1 + (\lambda_1 + \lambda_2) n_b \gamma}  \frac{1}{\tilde{n}_b}  )^H \big] ) \log{\frac{\sum_{b=1}^N n_b}{N \epsilon_D}} $
    }
    $H \leftarrow 0$, $H^{\text{cad}} \leftarrow 0$, $c \leftarrow +\infty$\\
    \For(){$h =  S_{b,t}, \ldots, 1$}{

        \If(){$ f(h) < c $}{
            $c \leftarrow f(h)$\\
            $H^{\text{cad}} \leftarrow h$\\

        }
    
        \If(){$ f(h) \leq C_{j,\text{bud}}$ $\textbf{and}$ $f(h+1) > C_{j,\text{bud}}$ }{
        $H \leftarrow h$\\
        }

    }
    \If(){$H = 0$}{
        $H \leftarrow H^{\text{cad}}$\\
    }
    \For(terminal iteration){$h = 0,\ldots,H$}{
    choose $i$ uniformly at random in local dataset\\
    find $\Delta \alpha_{b}^i$ of $\boldsymbol{\alpha}_{b}$ to maximize the local optimizing function \eqref{eq:local}\\
    $\Delta \boldsymbol{\alpha}_{[t]} \leftarrow (\Delta \boldsymbol{\alpha}_{[t]})^{i} + \Delta \alpha_{b}^i$\\
    $\boldsymbol{w}_{b,t} \leftarrow \boldsymbol{w}_{b,t} + \frac{1}{\lambda_1 + \lambda_2} \Delta \alpha_{b}^i \boldsymbol{x}_{b,t}^i$\\
    }
    \KwOut{$\Delta \boldsymbol{\alpha}_{[t]}$, $\Delta \boldsymbol{w}_{b,t} \triangleq \frac{1}{\lambda_1 + \lambda_2} \big(\frac{1}{S_{b,t}} \sum_{i=1}^{S_{b,t}} x_{b,t}^i \alpha_b^i \big) $}
\end{procedure}

\begin{theorem}\label{thm:thetamin}

    Assume that the loss function $\mathcal{L}$ is ($1/\gamma$)-smooth,
    for any convergence target $\epsilon_D$, in order to solve the MTL problem in \eqref{eq:problem}, if there exists a terminal iteration number $H_b$, $b \in \{ 1, \cdots,  N\}$  satisfying \eqref{eq:rH}, the problem is feasible.
\end{theorem}
\begin{proof}
    Considering the constraint in \eqref{eq:problem}, a direct analysis leads to that the iteration number $K$ under BS $B_b$ should satisfy:
    \begin{equation} \label{eq:resourceKH}
        K \leq \left\lfloor \min_{j\in \{1,\ldots,J\}} C_{j,\text{bud}} \bigg[ M \sum_{b=1}^N \big( C_{j,\text{BS}} + {N_b} H_b C_{j,\text{dev}} \big) \bigg]^{-1}  \right\rfloor.
    \end{equation}
    Taking account the convergence analysis in Theorem \ref{the},
    \begin{equation} \label{eq:resourceKL}
        K > \Big(1 - (1 - \Theta) \frac{\eta^*}{T^*} \Big) \log{\frac{\sum_{b=1}^N n_b}{N \epsilon_D}}.
    \end{equation}
    Rearranging \eqref{eq:resourceKH} and \eqref{eq:resourceKL}, we can derive that:
    \begin{equation}\label{eq:theta_constraint}
        \begin{split}
            &\Big(1 - (1 - \Theta) \frac{\eta^*}{T^*} \Big) \log{\frac{\sum_{b=1}^N n_b}{N \epsilon_D}} <  \\
            & \left\lfloor \min_{j\in \{1,\ldots,J\}} C_{j,\text{bud}} \bigg[ M \sum_{b=1}^N \big( C_{j,\text{BS}} + {N_b} H_b C_{j,\text{dev}} \big) \bigg]^{-1}   \right\rfloor.
        \end{split}
    \end{equation}
    Substituting \eqref{eq:Theta} into \eqref{eq:theta_constraint}, we have the theorem.
\end{proof}

\noindent\textit{Remark}: Theorem \ref{thm:thetamin} provides a feasible range of $H_b$, which will help determine an appropriate terminal iteration number. On the basis of that, the chosen BS iteration number will potentially reduce the overall resource consumption.

\begin{figure}[h]
    \centering
    \subfigure[Required $K$ for convergence]{
        \includegraphics[width=0.225\textwidth]{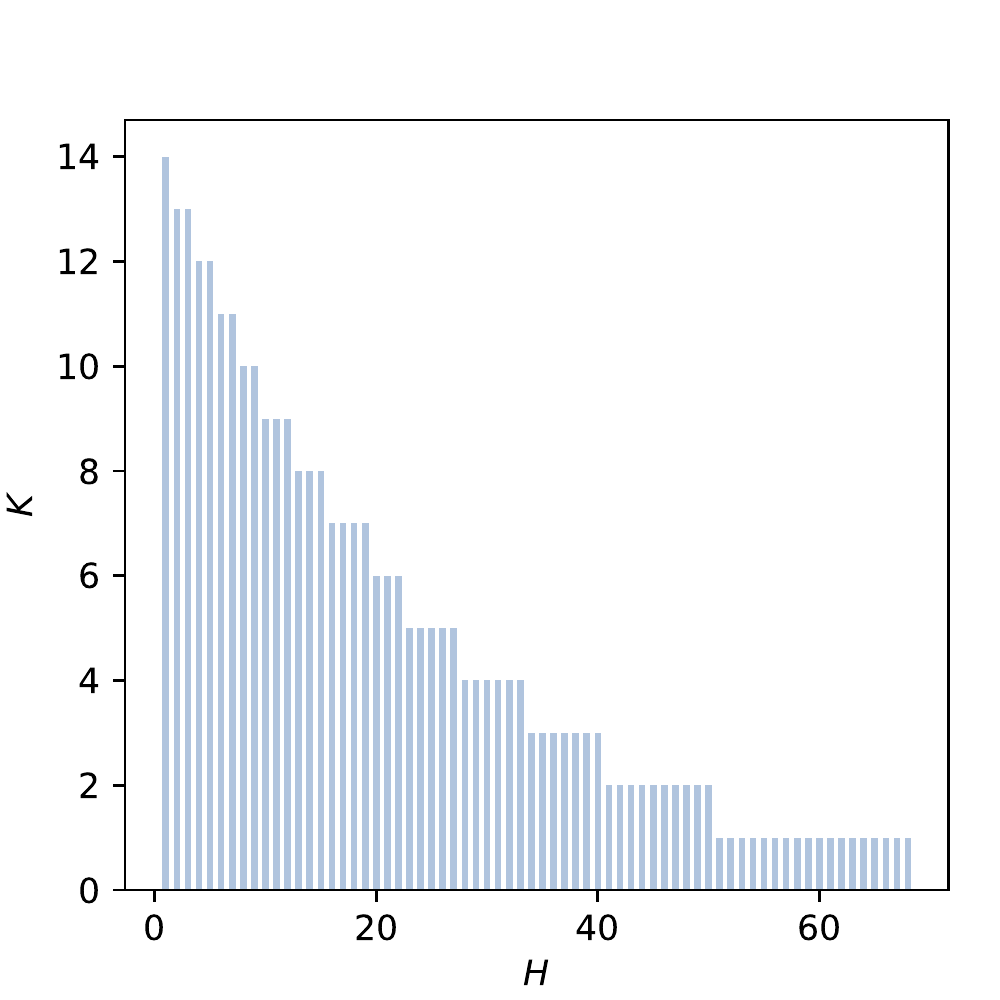}
    }
    \subfigure[$C_{j, \text{cost}}$ until the convergence target]{
        \includegraphics[width=0.225\textwidth]{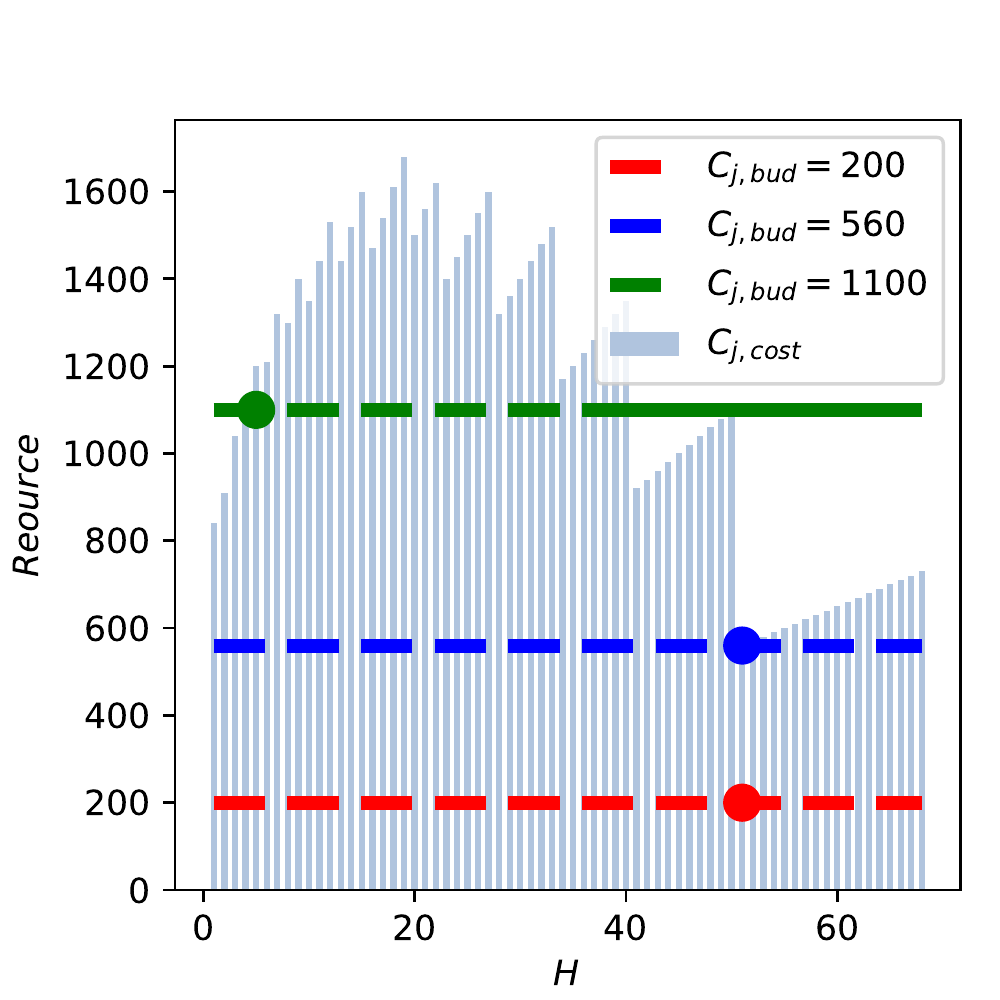}
    }
    \caption{BS iteration number $K$ and resource cost $C_{j,\text{cost}}$ until convergence target $\epsilon_D$ over different settings of $H_b = H$, $ b \in \{ 1, \ldots, N\} $.}
    \label{fg:cost}
\end{figure}

As given in Theorem \ref{thm:thetamin}, the left side of \eqref{eq:rH} represents the needed resource cost to meet the convergence target $\epsilon_D$ \eqref{eq:duality_gap} under certain number of terminal iterations.
In case where $N=5$, $N_b=5$, $\epsilon_D=0.01$, we demonstrate in Fig.~\ref{fg:cost} the numerical relationship between the BS iteration number and the resource cost with respect to a unanimous terminal iteration number $H_b = H$, $ b \in \{ 1, \ldots, N\} $.
Accordingly, with limited resource budget $C_{j,\text{bud}}$, there exists three possible cases and they are represented by horizontal lines of different colors in Fig.~\ref{fg:cost} where the solid lines means $C_{j,\text{bud}} > C_{j, \text{cost}}$ while dashed lines means $C_{j,\text{bud}} < C_{j, \text{cost}}$. Besides, the dot indicates the value of chosen $H$.
\begin{enumerate}
\item The given resource budget $C_{j,\text{bud}}$ is too limited to converge, just as shown as the red line in Fig.~\ref{fg:cost}. In this case, RHFedMTL capably traverses all the range of $H_b$ to minimize the resource cost of terminals (i.e., to be most resource efficient) and reduce the number of BS iterations to meet the limited resource budget.
\item The given resource budget $C_{j,\text{bud}}$ is just enough to converge, just as shown as the blue line in Fig.~\ref{fg:cost}. In this case, RHFedMTL capably traverses all the feasible range of $H_b$ to satisfy \eqref{eq:rH}, making a balance between learning speed and efficiency. Correspondingly, it sets the BS iteration number according to \eqref{eq:k}.
\item The given resource budget $C_{j,\text{bud}}$ leads to a resource surplus, just as shown as the green line in Fig.~\ref{fg:cost}. In this case, RHFedMTL can choose a appropriately smaller valid number of $H_b$, so as to fully leverage the resource budget. Notably, the adaption of $H_b$ also implies the ability to cope with the straggler issue.
\end{enumerate}

Finally, under the RHFedMTL framework in Algorithm \ref{al:rhfedmtl}, we present the \textsc{LocalDualMethod} in Procedure \ref{localresource}.

\section{Simulation and Numerical Results}\label{Simulation}

In this part, we illustrate the performance of the proposed RHFedMTL algorithm.
Consistent with Mocha \cite{mocha}, we adopt a dataset consisting of mobile phone accelerometer and gyroscope data collected from 30 individuals for human activity recognition \cite{activity}.
Furthermore, we attempt to classify the sitting gesture and the other activities from the provided feature vectors (i.e., time and frequency domain variables), and consider the data from different individuals as different tasks. Specifically each task contains $70$ pieces of data. In other words, there exist at most $30$ tasks, and more tasks could bring enhanced regulation as in \eqref{eq:r} and \eqref{eq:r2}. Meanwhile, the training dataset of one task is distributed among terminals while the test dataset is stored in BS for performance evaluation.

\begin{table}[b]
    \centering
    \caption{Notations and Default Settings.}
    \label{tb:simNote}
    \begin{tabular}{lrr}
        \toprule
        Notations           & Description                 & Default Setting \\
        \midrule
        \textit{H}          & No. of Terminal Iterations   & $2$               \\
        $C_{j, \text{bud}}$ & Resource Budget             & $1,400$            \\
        $C_{j, \text{dev}}$ & Cost for Terminal Iteration & $0.1$             \\
        $C_{j, \text{BS}}$  & Cost for BS Iteration       & $10$              \\
        $N$                 & No. of Tasks                 & $5$    \\
        $N_b$               & No. of Terminals Under One Task& $5$\\
        $\gamma$            & The Smoothness of Loss Function & $1$\\
        $\lambda_1$         & Weight of Self Regulation& $10^{-4}$\\
        $\lambda_2$ & Weight of Multi-Task Regulation& $10^{-6}$\\
        \bottomrule
    \end{tabular}
\end{table}

\begin{figure}[htb]
    \includegraphics[width=0.45\textwidth]{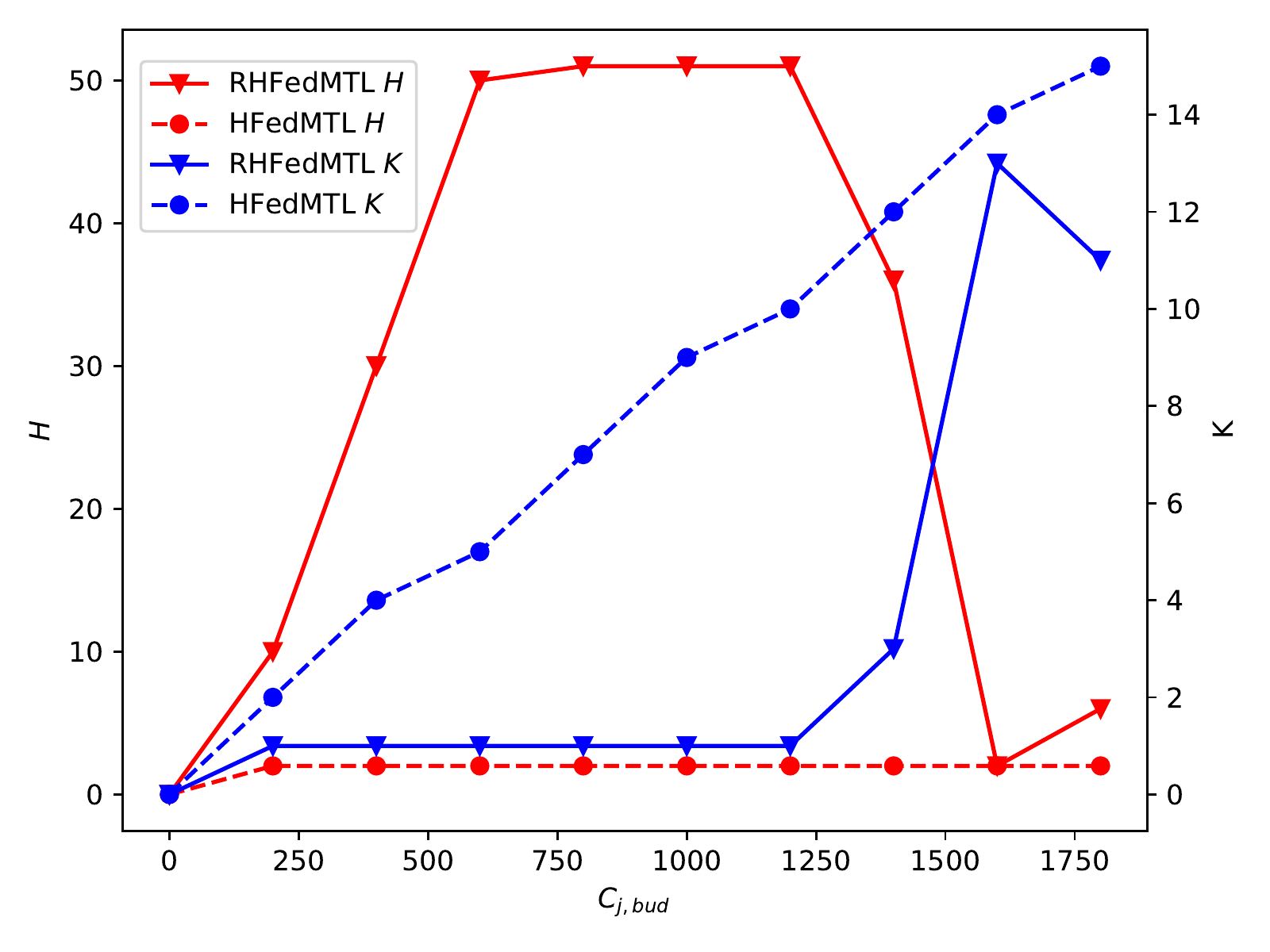}
    \caption{Chosen $H$ and $K$ under different resource budget $C_{j,\text{bud}}$.}
    \label{fg:HK}
\end{figure}

\begin{figure}[htb]
    \centering
    \subfigure[$C_{j,\text{bud}} = 400$]{
        \includegraphics[width=0.22\textwidth]{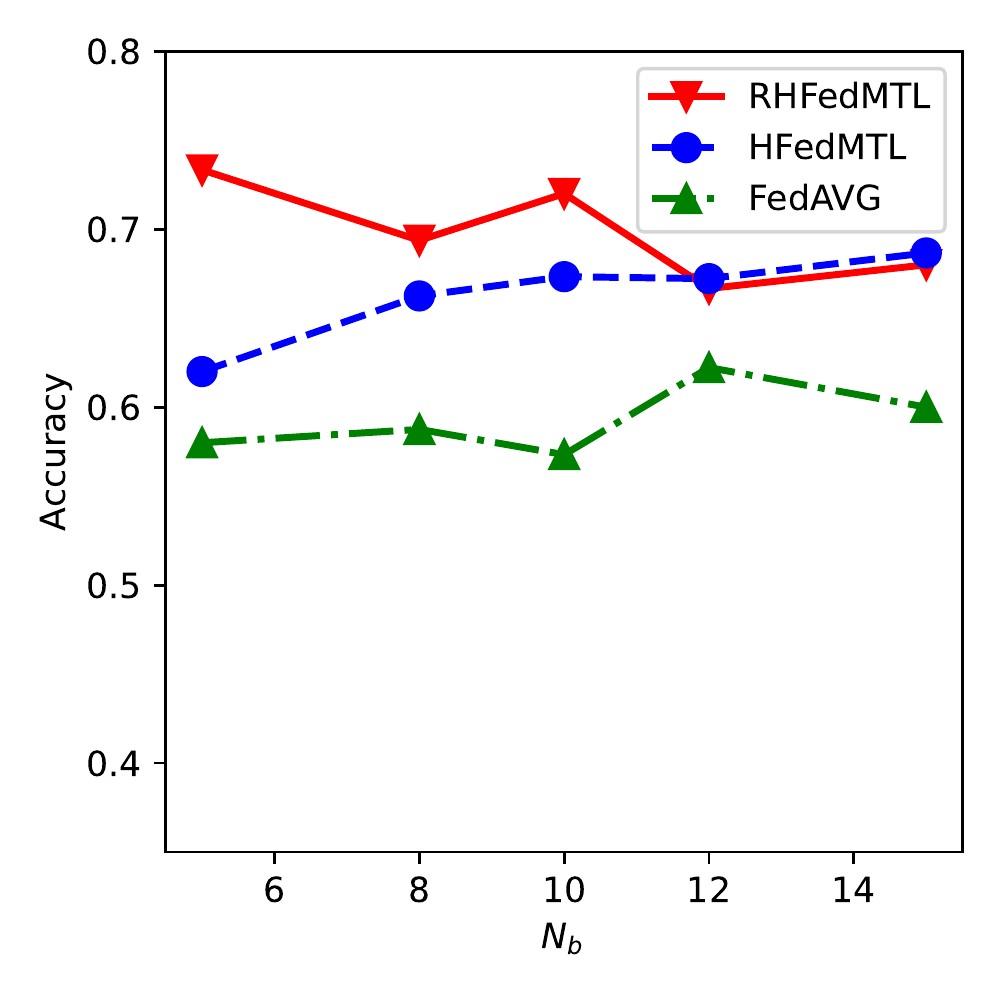}
    }
    \subfigure[$C_{j,\text{bud}} = 1,400$]{
        \includegraphics[width=0.22\textwidth]{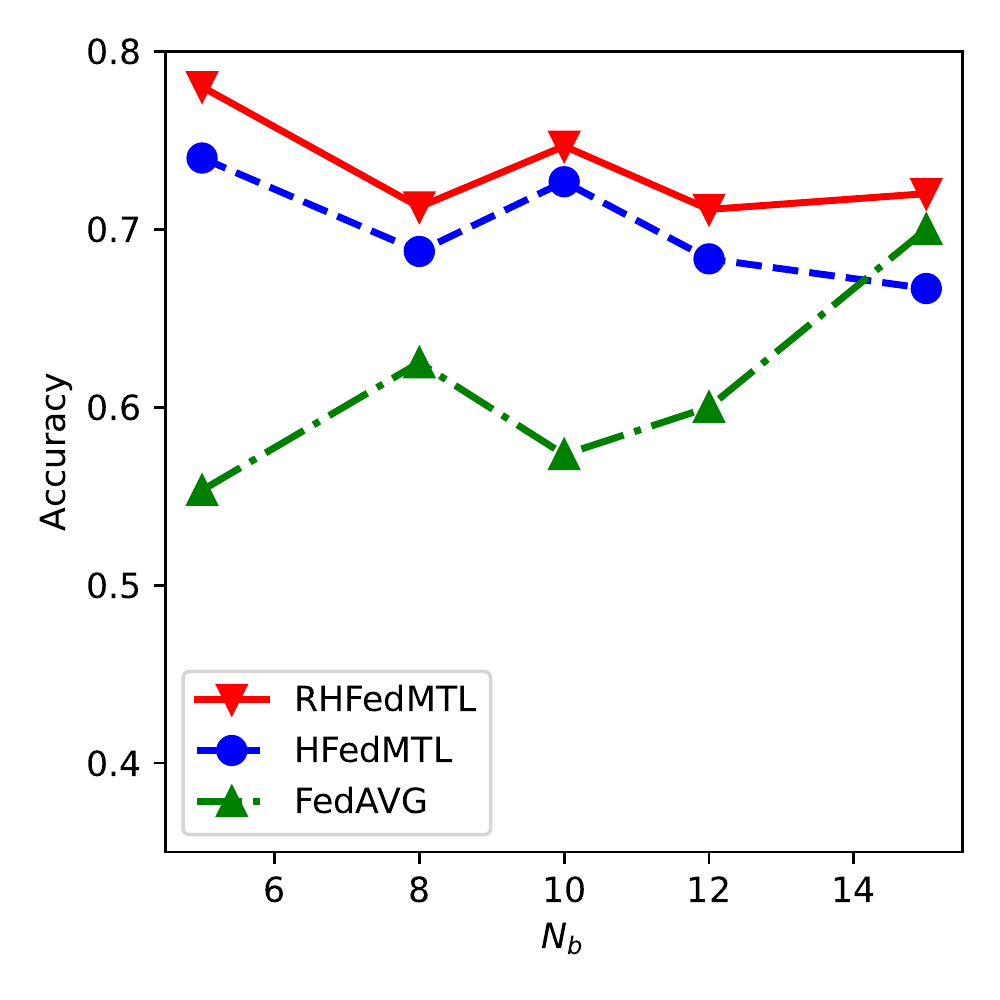}
    }
    \caption[short]{Model accuracy versus different terminal number $N_b$.}
    \label{fg:terminal}
\end{figure}

\begin{figure}[htb]
    \centering
    \subfigure[$C_{j,\text{bud}} = 400$]{
        \includegraphics[width=0.22\textwidth]{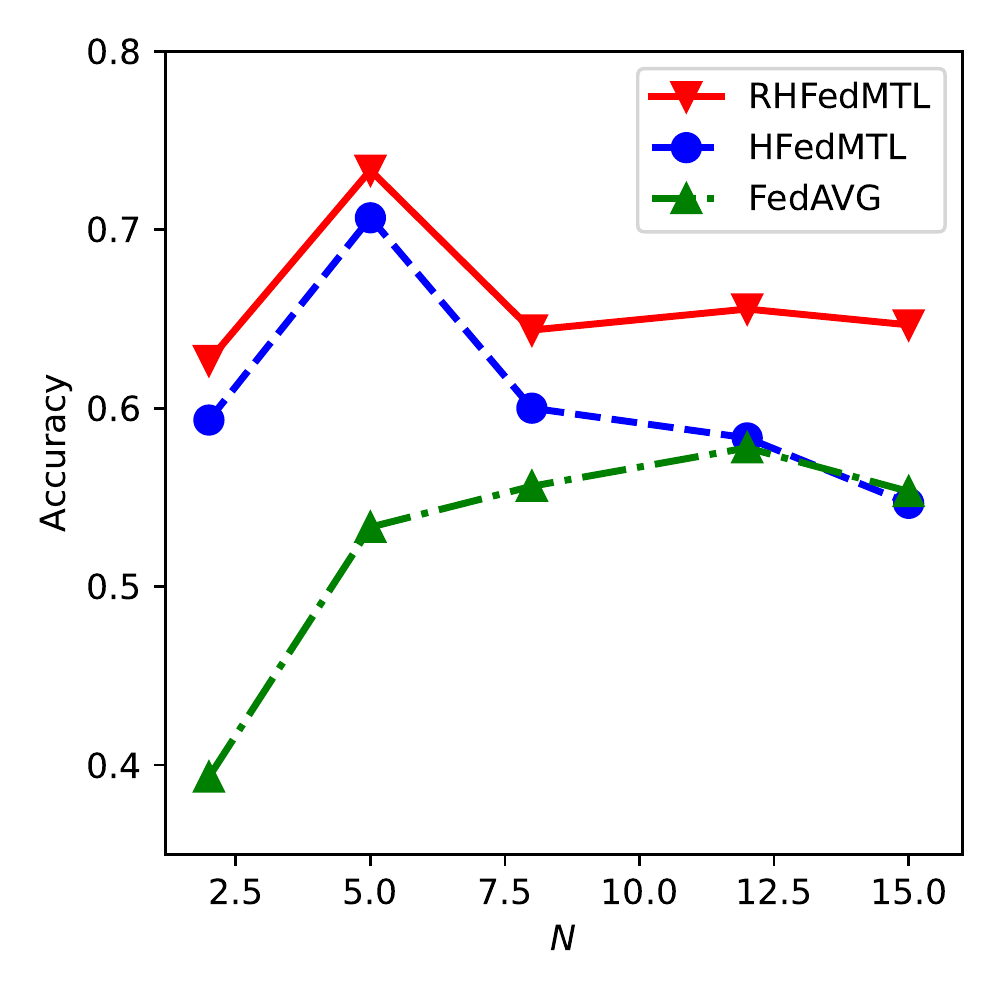}
    }
    \subfigure[$C_{j,\text{bud}} = 1,400$]{
        \includegraphics[width=0.22\textwidth]{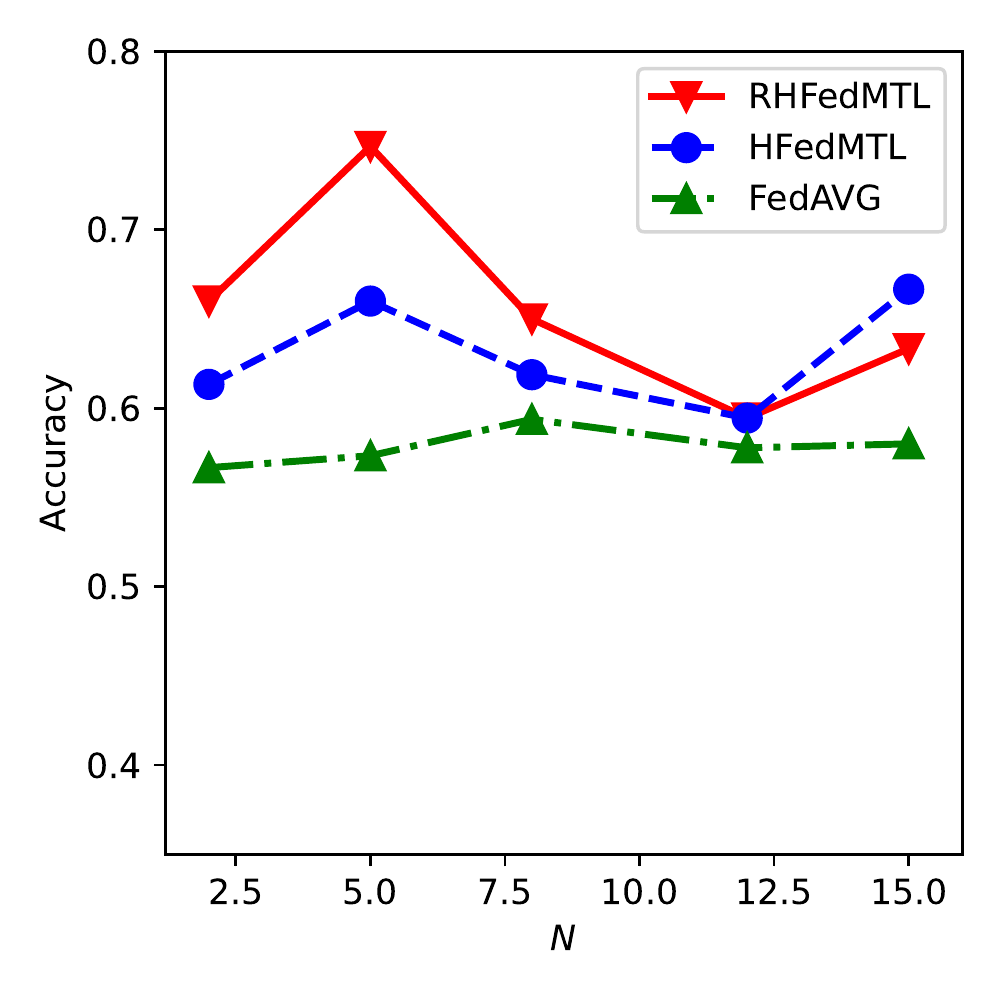}
    }
    \caption[short]{Model accuracy with respect to different task number $N$.}
    \label{fg:N}
\end{figure}

\begin{figure}[htb]
    \centering
    \subfigure[$C_{j,\text{bud}} = 400$]{
        \includegraphics[width=0.22\textwidth]{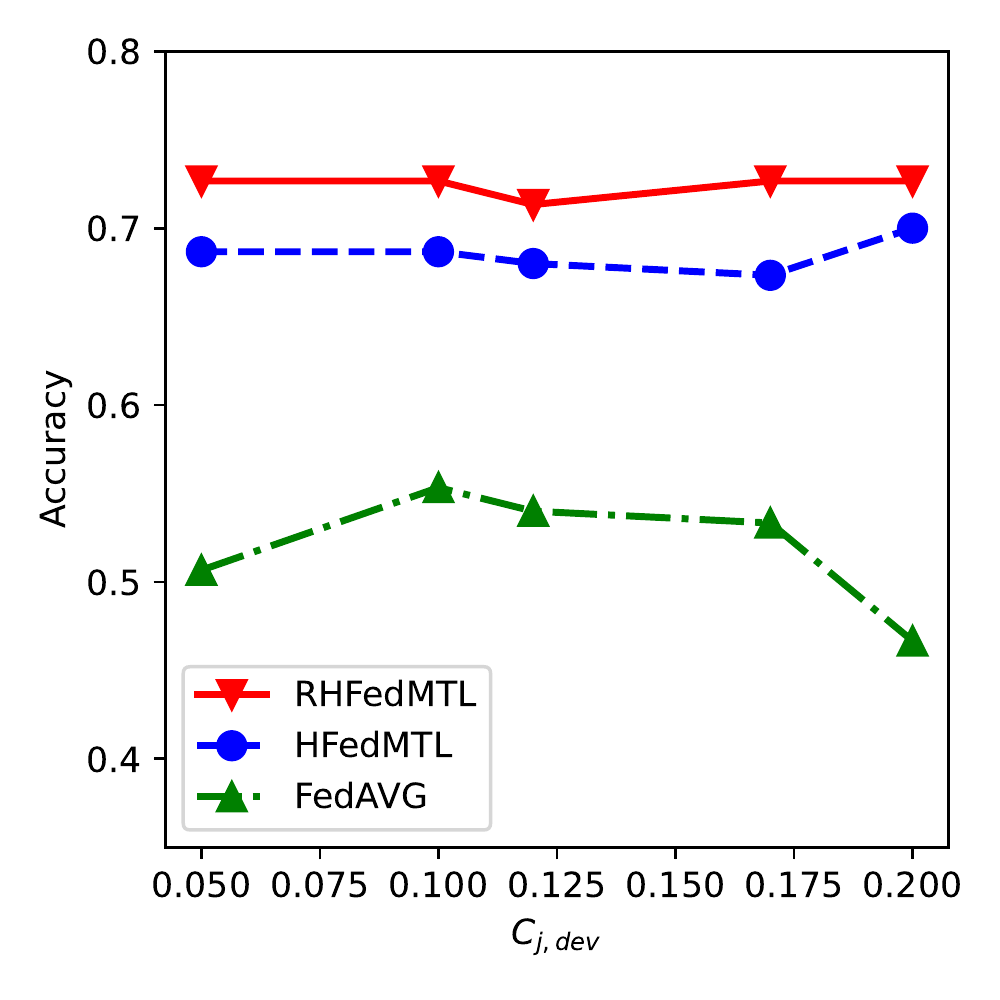}
    }
    \subfigure[$C_{j,\text{bud}} = 1,400$]{
        \includegraphics[width=0.22\textwidth]{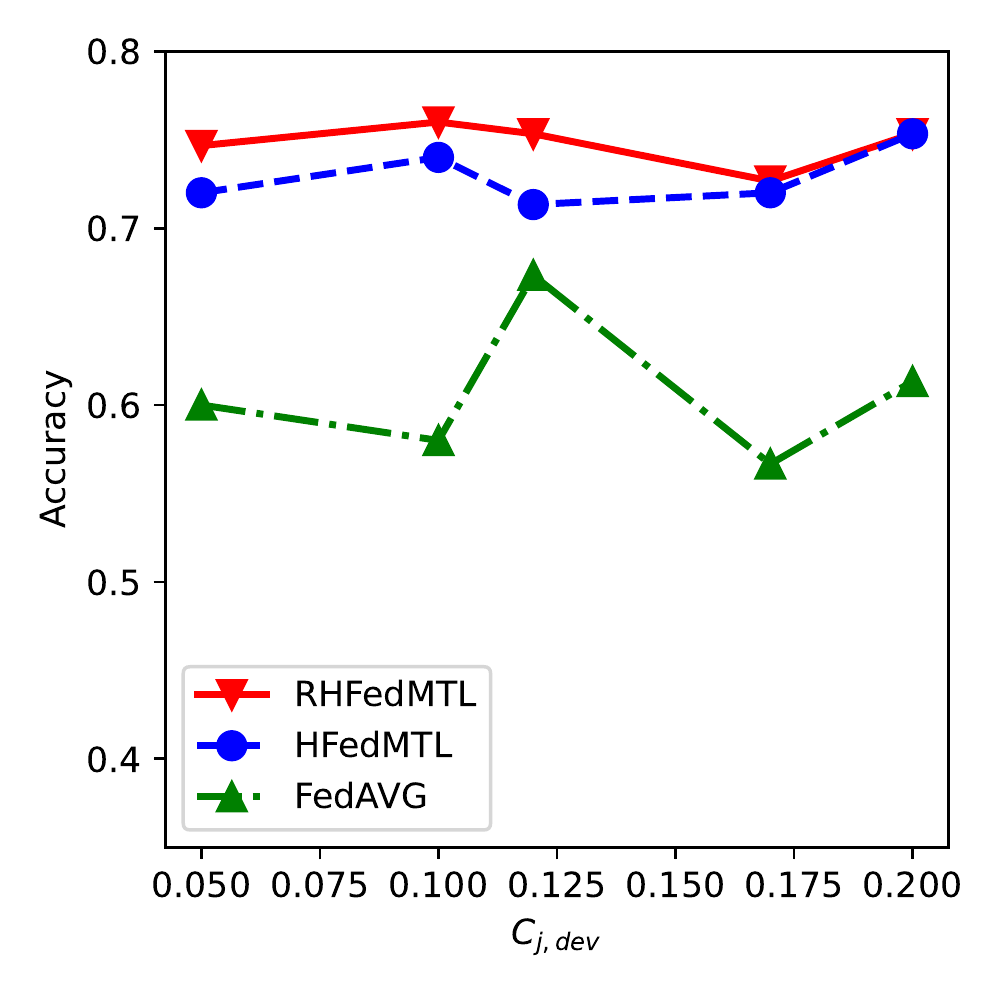}
    }
    \caption[short]{Model accuracy over terminal iteration cost $C_{j,\text{dev}}$.}
    \label{fg:Cdev}
\end{figure}

\begin{figure*}[htb]
    \centering
    \subfigure[RHFedMTL]{
        \includegraphics[width=0.31\textwidth]{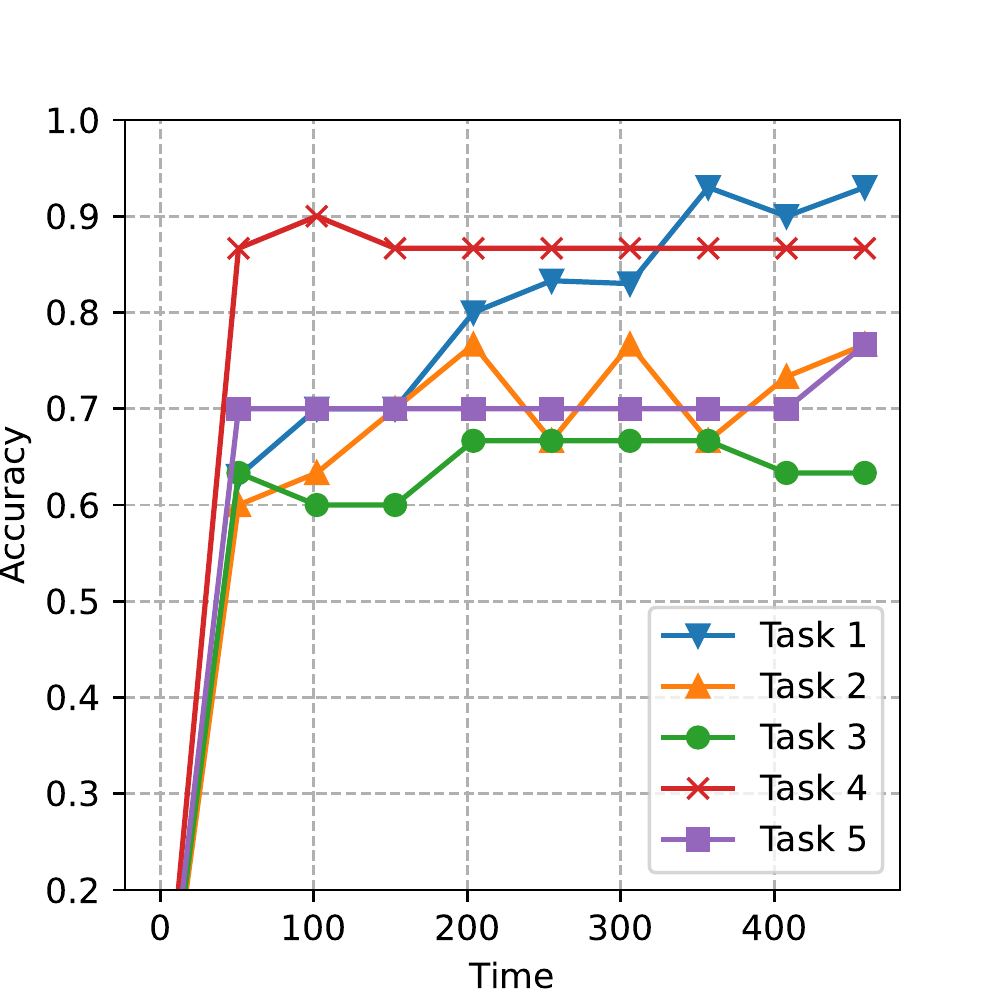}
    }
    \subfigure[HFedMTL]{
        \includegraphics[width=0.31\textwidth]{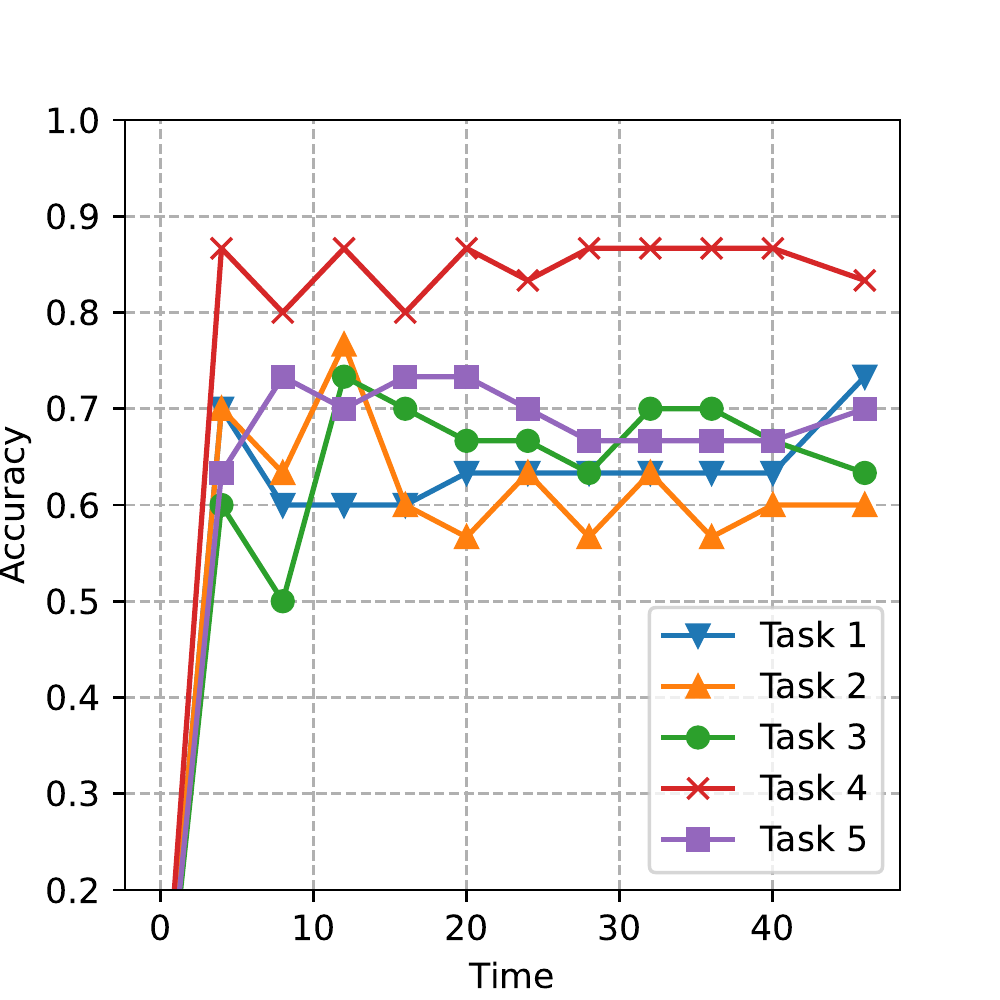}
    }
    \subfigure[FedAVG]{
        \includegraphics[width=0.31\textwidth]{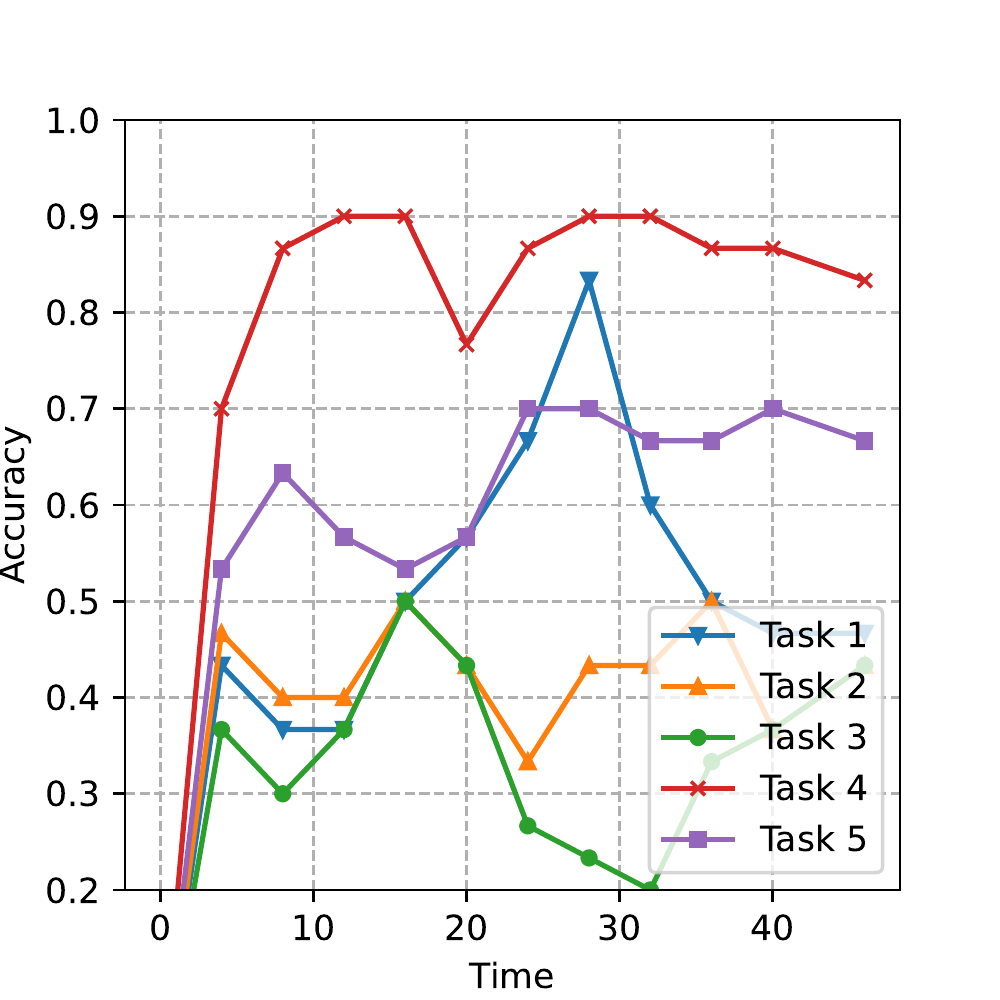}
    }
    \caption[short]{Task accuracy with $N = 5$, $C_{j,\text{bud}} = 1,400$.}
    \label{fg:acc}
\end{figure*}
\begin{table*}[htb]
    \centering
    \caption{Model accuracy with different resource budget under different local terminals.}
    \label{tb:acc}
    \begin{tabular}{|c|ccc|ccc|ccc|}
        \toprule
                           & \multicolumn{3}{c|}{RHFedMTL}           & \multicolumn{3}{c|}{HFedMTL}    & \multicolumn{3}{c|}{FedAVG}                                                                                                                                                                 \\ \hline
        $C_{j,\text{bud}}$ & \multicolumn{1}{c|}{\textit{$N_b = 5$}} & \multicolumn{1}{c|}{$N_b = 10$} & $N_b = 15$                  & \multicolumn{1}{c|}{$N_b = 5$} & \multicolumn{1}{c|}{$N_b = 10$} & $N_b = 15$ & \multicolumn{1}{c|}{$N_b = 5$} & \multicolumn{1}{c|}{$N_b = 10$} & $N_b = 15$ \\ \midrule
        $200$                & \multicolumn{1}{c|}{$0.70$}               & \multicolumn{1}{c|}{$0.70$}       & $0.68$                        & \multicolumn{1}{c|}{$0.68$}      & \multicolumn{1}{c|}{$0.65$}       & $0.67$       & \multicolumn{1}{c|}{$0.54$}      & \multicolumn{1}{c|}{$0.55$}       & $0.58$       \\ \hline
        $400$                & \multicolumn{1}{c|}{$0.73$}               & \multicolumn{1}{c|}{$0.72$}       & $0.68$                        & \multicolumn{1}{c|}{$0.62$}      & \multicolumn{1}{c|}{$0.67$}       & $0.69$       & \multicolumn{1}{c|}{$0.58$}      & \multicolumn{1}{c|}{$0.57$}       & $0.60$       \\ \hline
        $600$                & \multicolumn{1}{c|}{$0.72$}               & \multicolumn{1}{c|}{$0.71$}       & $0.70$                        & \multicolumn{1}{c|}{$0.63$}      & \multicolumn{1}{c|}{$0.72$}       & $0.70$       & \multicolumn{1}{c|}{$0.54$}      & \multicolumn{1}{c|}{$0.53$}       & $0.64$      \\ \hline
        $800$                & \multicolumn{1}{c|}{$0.72$}               & \multicolumn{1}{c|}{$0.72$}       & $0.70$                        & \multicolumn{1}{c|}{$0.71$}      & \multicolumn{1}{c|}{$0.70$}       & $0.71$       & \multicolumn{1}{c|}{$0.61$}      & \multicolumn{1}{c|}{$0.53$}       & $0.63$       \\ \hline
        $1,000$               & \multicolumn{1}{c|}{$0.73$}               & \multicolumn{1}{c|}{$0.72$}       & $0.71$                        & \multicolumn{1}{c|}{$0.65$}      & \multicolumn{1}{c|}{$0.71$}       & $0.68$       & \multicolumn{1}{c|}{$0.59$}      & \multicolumn{1}{c|}{$0.55$}       & $0.62$       \\ \hline
        $1, 200$               & \multicolumn{1}{c|}{$0.77$}               & \multicolumn{1}{c|}{$0.73$}       & $0.71$                        & \multicolumn{1}{c|}{$0.68$}      & \multicolumn{1}{c|}{$0.67$}       & $0.70$       & \multicolumn{1}{c|}{$0.56$}      & \multicolumn{1}{c|}{$0.54$}       & $0.64$       \\ \hline
        $1,400$               & \multicolumn{1}{c|}{$0.78$}               & \multicolumn{1}{c|}{$0.75$}       & $0.72$                        & \multicolumn{1}{c|}{$0.74$}      & \multicolumn{1}{c|}{$0.73$}       & $0.67$       & \multicolumn{1}{c|}{$0.55$}      & \multicolumn{1}{c|}{$0.57$}       & $0.70$       \\ \hline
        $1,600$               & \multicolumn{1}{c|}{$0.74$}               & \multicolumn{1}{c|}{$0.71$}       & $0.68$                        & \multicolumn{1}{c|}{$0.70$}      & \multicolumn{1}{c|}{$0.72$}       & $0.69$       & \multicolumn{1}{c|}{$0.55$}      & \multicolumn{1}{c|}{$0.60$}       & $0.72$       \\ \bottomrule
    \end{tabular}
\end{table*}

In order to express the overall performance of the system more intuitively,
we calculate the overall accuracy by averaging the accuracy of all separate tasks on their test dataset under various settings, and then dive into the real-time task performance of all separate tasks lately in Fig.~\ref{fg:acc}.
For ease of representation and interpretation of results, we primarily consider $J= 1 $ type of resources and consider the energy as the single resource type in our experiments. The resource cost includes computational cost for terminal iteration (i.e., local training process) and BS iteration (i.e., aggregation and transmission), consistent with \eqref{eq:problem}.
We compare RHFedMTL with vanilla HFedMTL and classical methods including FedAVG~\cite{fedavg}. 
Since RHFedMTL can adjust its terminal iteration number while others can not, we default the terminal iteration number $H$ for HFedMTL and FedAVG to be 2.
The system parameters are summarized in Table \ref{tb:simNote}.

We first experiment with the relationship between terminal iteration number $H_b$ and BS iteration number $K$ needed to reach target duality gap in \eqref{eq:duality_gap}.
For simplicity of representation, we set $H \triangleq \min_b{H_b}$ as shown in \eqref{eq:rH} and set $H_b$ to be the same as $H$ for all BSs.
The experiment results are shown in Fig. \ref{fg:cost}.
Consistent with our discussions above (e.g., \eqref{eq:Theta} and \eqref{eq:rH}), the terminal iteration number $H$ and BS iteration number $K$ are negatively correlated while the relationship between $H$ and system resource cost is not always monotonic.
Therefore, if the value of $C_{j, \text{BS}}$ is relatively large, the reduction of $K$ can compensate the resource cost due to the increased the terminal iteration number $H$. 
Numerical results of chosen $H$ and $K$ under different resource budget $C_{j,\text{dev}}$ is shown in Fig. \ref{fg:HK}.
Due to the lack of the capability for resource awareness, HFedMTL adopts a fixed terminal iteration number and can only increase the BS iteration number to drain up the resource budget. But, as further verified in Fig.~\ref{fg:terminal} and Table \ref{tb:acc}, RHFedMTL can adjust $H$ and $K$ simultaneously to reach better performance under given resource budget. Moreover, given a limited resource budget, it can be observed from Fig. \ref{fg:HK}, RHFedMTL gives priority to increase the terminal iteration number since it costs less resources
and gradually increases the BS iteration number to yield superior performance as the resource budget grows.

As for Fig.~\ref{fg:terminal}, since the dataset under one task is limited, the number of data stored in one terminal decreases as terminal number $N_b$ grows.
Therefore, with a smaller local dataset size for terminals, RHFedMTL leads to a relatively smaller range of $H$ and produces performance similar to HFedMTL (i.e., similar $H$ and $K$) in Fig.~\ref{fg:HK}, it can be observed from Fig.~\ref{fg:terminal} and Table \ref{tb:acc} that the performance gap between HFedMTL and RHFedMTL decreases as $N_b$ grows. Similarly, under a resource-abundant case, the performance gap of RHFedMTL and HFedMTL will also decrease since resource budget is no longer the bottleneck to inhibit the learning performance.

In Fig.~\ref{fg:N}, we compare the model accuracy with different task numbers (i.e., different BS numbers).
Benefiting from the multi-task regulation, the performance of RHFedMTL and HFedMTL increase along with the increase of the number of tasks.
However, a over-larger task number will diminish the performance gain, given the complication to learn the relationship among multiple tasks. On the other hand, Fig. \ref{fg:Cdev} examines the impact of terminal iteration cost $C_{j,\text{dev}}$.
Consistent with our intuition, the excessive cost of terminal iteration $C_{j,\text{dev}}$ will decrease the model accuracy.
From the figures above, RHFedMTL and HFedMTL shows their robustness under various settings and RHFedMTL demonstrates their superiority.

In Fig.~\ref{fg:acc}, we plot the accuracy of each task separately.
In this simulation, the resource budget $C_{j,\text{bud}}$ is $1,400$ while the task number $N$ is $5$.
Since the resource budget is limited, RHFedMTL chooses a bigger terminal iteration number $H$ to make the system more cost-efficient.
Therefore, the BS iteration number $K$ of RHFedMTL will be smaller accordingly and leads to fewer data points as shown in Fig.~\ref{fg:acc}, as the task accuracy is evaluated after each BS iteration.
Again, RHFedMTL and HFedMTL can always reach a higher accuracy than FedAVG and facilitate the learning of other tasks.
As already discussed, compared with HFedMTL, RHFedMTL yields better performance.

Finall, we conduct more performance sensitivity analysis and testify the model accuracy of RHFedMTL with respect to different regulation parameter $\lambda_1$ and $\lambda_2$. As shown in Fig.~\ref{fg:weight},
 for the self-regulation parameter $\lambda_1$, the model accuracy will decrease, since it will loss the functionality of self regulation if we set the value of $\lambda_1$ too trivial; while a over-large value of $\lambda_1$ can potentially degrade the performance as well. 
Similar phenomena can also be observed in results corresponding to multi-task regulation parameter $\lambda_2$. Therefore, we set the default value of $\lambda_1$ and $\lambda_2$ to be $10^{-4}$ and $10^{-6}$, respectively.

\begin{figure}[tb]
    \centering
    \subfigure[$\lambda_1$]{
        \includegraphics[width=0.22\textwidth]{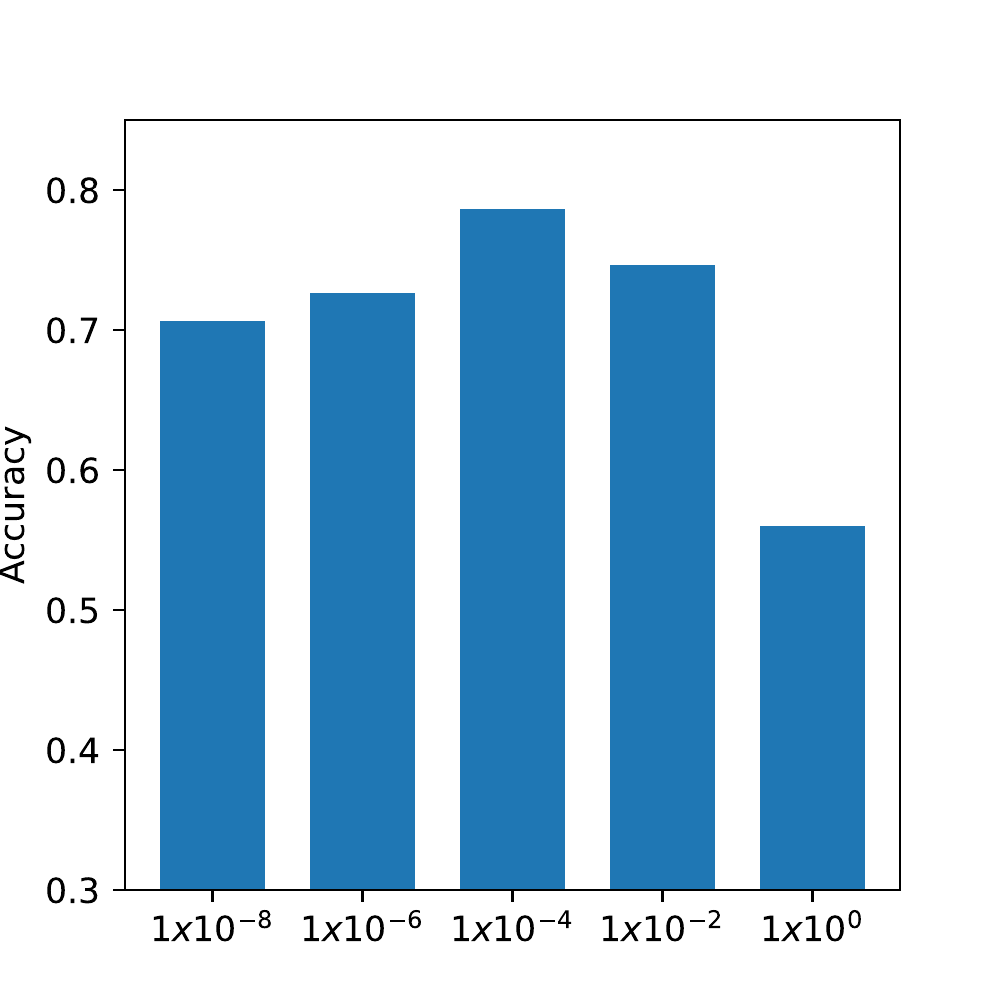}
    }
    \subfigure[$\lambda_2$]{
        \includegraphics[width=0.22\textwidth]{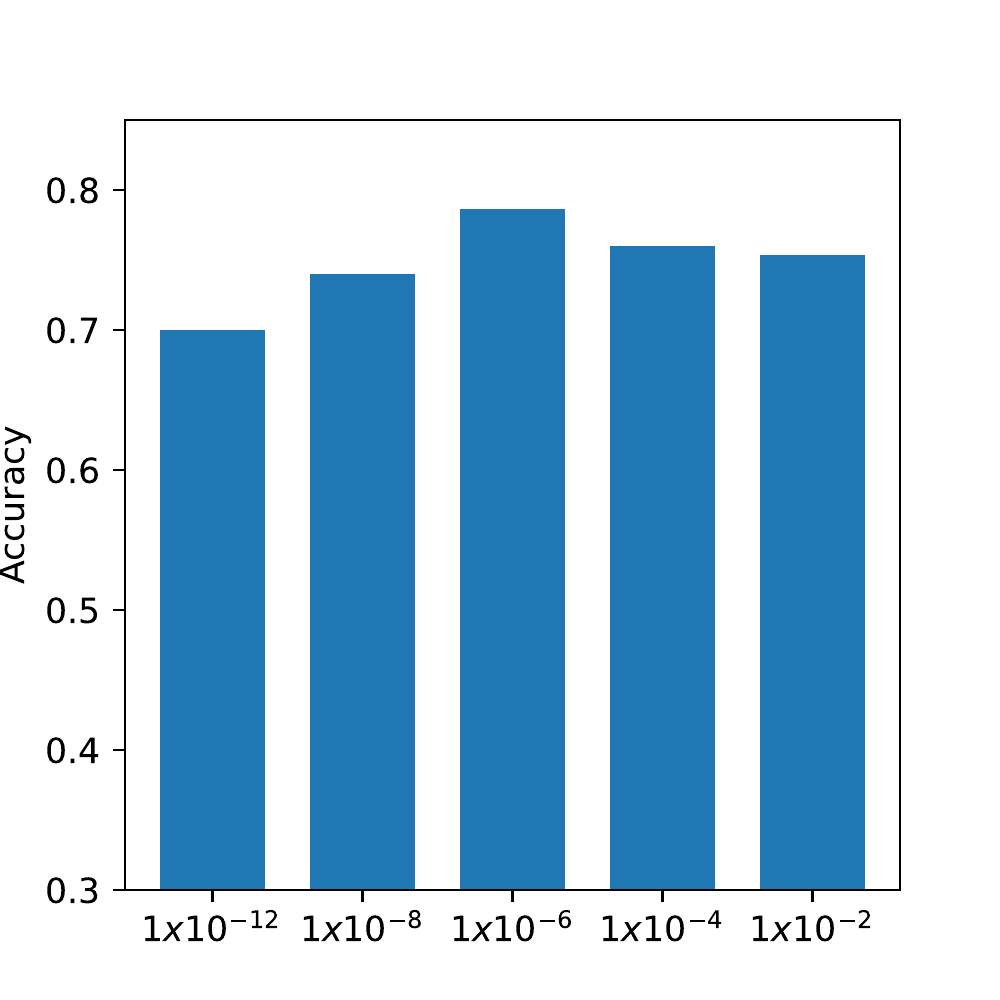}
    }
    \caption[short]{Model accuracy of RHFedMTL with different settings of regulation parameters.}
    \label{fg:weight}
\end{figure}
\section{Conclusion}
In this paper, we have addressed the importance of network AI and proposed an RHFedMTL framework based on the primal-dual method towards tackling federated MTL problems with stragglers. In particular, 
RHFedMTL has taken the hierarchy of cellular works into consideration and encompassed a three-tier iteration mechanism including server iteration, BS iteration and terminal iteration.
Moreover, the primal-dual method SDCA has been leveraged to effectively transform the coupled multi-task learning into some local optimization problems within BSs.
We have analyzed the convergence bound of the proposed framework, and derived a guiding relationship between terminal iteration and BS iteration. Afterwards, we have developed a resource-aware learning strategy for local terminals and BSs to obtain more satisfactory learning performance under given resource budget.
Extensive experimentation results have demonstrated the effectiveness and robustness of the proposed method.

\appendix[Proof of Lemma \ref{lemmaDt}]
    \label{proof_lemmaDt}
    \begin{proof}
        Recalling the relationship $y_{b}^i \triangleq \boldsymbol{w}^\intercal_{b} \boldsymbol{x}_{b}^i$ in \eqref{eq:mtl0}, the Lagrangian dual of problem could be derived as follows:
        \begin{align}
              & D (\boldsymbol{\alpha} ) \notag                                                                                                                                                                                                                                                                              \\
            = & \inf_{\mathbf{W}, y_{b}^i} \frac{1}{N} \sum_{b=1}^N  \bigg\{ \frac{\lambda_1}{2} \| \boldsymbol{w}_{b} \Vert^2 + \frac{\lambda_2}{2} \left\lVert \boldsymbol{w}_{b}  - \boldsymbol{r}_b\right\rVert^2  \bigg. \notag                                                                                         \\
              & \bigg.  + \frac{1}{n_b}\sum_{i=1}^{n_b} \Big( \mathcal{L}(y_{b}^i) + {\alpha}_{b}^i(y_{b}^i-\boldsymbol{w}_{b} ^\intercal \boldsymbol{x}_{b} ^i) \Big) \bigg\}  \notag                                                                                                                                               \\
            = & \inf_{y_{b,t}^i} \frac{1}{N} \sum_{b=1}^N  \frac{1}{n_b} \sum_{i=1}^{n_b} \Big( \mathcal{L}(y_{b}^i) +  {\alpha}_{b}^i y_{b}^i \Big)  + \inf_\mathbf{W} \frac{1}{N} \sum_{b=1}^N  \notag                                                                                                                     \\
              & \Big(  \frac{\lambda_1}{2} \| \boldsymbol{w}_{b} \Vert^2 + \frac{\lambda_2}{2} \left\lVert \boldsymbol{w}_{b}  - \boldsymbol{r}_b\right\rVert^2 - \frac{1}{n_b} \sum_{i=1}^{n_b} {\alpha}_{b}^i \boldsymbol{w}_{b} ^\intercal \boldsymbol{x}_{b} ^i \Big) \notag                                                     \\
            = & - \frac{1}{N} \sum_{b=1}^N  \frac{1}{n_b} \sum_{i=1}^{n_b} \sup_{y_{b}^i} \Big(- {\alpha}_{b}^i y_{b,t}^i - \mathcal{L}(y_{b}^i) \Big) + \inf_\mathbf{W} \frac{1}{N} \sum_{b=1}^N  \notag                                                                                                                    \\
              & \Big( \frac{\lambda_1}{2} \| \boldsymbol{w}_{b} \Vert^2 + \frac{\lambda_2}{2} \left\lVert \boldsymbol{w}_{b}  - \boldsymbol{r}_b\right\rVert^2  - \frac{1}{n_b} \sum_{i=1}^{n_b} {\alpha}_{b}^i \boldsymbol{w}_{b} ^\intercal \boldsymbol{x}_{b} ^i \Big)  \notag                                                    \\
            = & \frac{1}{N} \sum_{b=1}^N  \frac{1}{n_b} \sum_{i=1}^{n_b} -\mathcal{L}^*(- {\alpha}_{b}^i ) + \notag                                                                                                                                                                                                          \\
              & \inf_\mathbf{W} \frac{1}{N}   \sum_{b=1}^N  \Big( \frac{\lambda_1}{2} \| \boldsymbol{w}_{b} \Vert^2 + \frac{\lambda_2}{2} \left\lVert \boldsymbol{w}_{b}  - \boldsymbol{r}_b\right\rVert^2 - \frac{1}{n_b} \sum_{i=1}^{n_b} {\alpha}_{b}^i \boldsymbol{w}_{b} ^\intercal \boldsymbol{x}_{b} ^i \Big) \label{eq:dual}
        \end{align}
        Since the optimization of $\boldsymbol{w}_b$ in \eqref{eq:dual} is independent from the models in other tasks, after taking the gradient of $\boldsymbol{w}_b$ in the second term and setting it to zero, for any model $\boldsymbol{w}_b$ we would have:
        \begin{equation}\label{eq:wt}
            \begin{split}
                \boldsymbol{w}_{b}  &= \frac{1}{\lambda_1 + \lambda_2} \Big( \lambda_2 \boldsymbol{r}_b + \frac{1}{n_b} \sum_{i=1}^{n_b} {\alpha}_{b}^i \boldsymbol{x}_{b} ^i \Big). \\
            \end{split}
        \end{equation}
        Substituting the chosen $\boldsymbol{w}_b$ in \eqref{eq:wt} into \eqref{eq:dual}, we get the lemma.
    \end{proof}
\bibliographystyle{IEEEtran}
\bibliography{cite}

\begin{thebibliography}{10}
\providecommand{\url}[1]{#1}
\csname url@samestyle\endcsname
\providecommand{\newblock}{\relax}
\providecommand{\bibinfo}[2]{#2}
\providecommand{\BIBentrySTDinterwordspacing}{\spaceskip=0pt\relax}
\providecommand{\BIBentryALTinterwordstretchfactor}{4}
\providecommand{\BIBentryALTinterwordspacing}{\spaceskip=\fontdimen2\font plus
\BIBentryALTinterwordstretchfactor\fontdimen3\font minus
  \fontdimen4\font\relax}
\providecommand{\BIBforeignlanguage}[2]{{%
\expandafter\ifx\csname l@#1\endcsname\relax
\typeout{** WARNING: IEEEtran.bst: No hyphenation pattern has been}%
\typeout{** loaded for the language `#1'. Using the pattern for}%
\typeout{** the default language instead.}%
\else
\language=\csname l@#1\endcsname
\fi
#2}}
\providecommand{\BIBdecl}{\relax}
\BIBdecl

\bibitem{yi_hfedmtl_2022}
X.~Yi, R.~Li, C.~Peng, J.~Wu, and Z.~Zhao, ``{{HFedMTL}}: {{Hierarchical}}
  federated multi-task learning,'' in \emph{Proc. {{IEEE PIMRC}} 2022},
  {Virtual Edition}, Sep. 2022.

\bibitem{6Gvisions}
L.~Zhang, Y.-C. Liang, and D.~Niyato, ``{6G visions: Mobile ultra-broadband,
  super Internet-of-Things, and artificial intelligence},'' \emph{China
  Communications}, vol.~16, no.~8, pp. 1--14, 2019.

\bibitem{intelligent5G}
R.~Li, Z.~Zhao, X.~Zhou \emph{et~al.}, ``{Intelligent 5G: When cellular
  networks meet artificial intelligence},'' \emph{IEEE Wireless
  Communications}, vol.~24, no.~5, pp. 175--183, 2017.

\bibitem{collective}
R.~Li, Z.~Zhao, X.~Xu \emph{et~al.}, ``{The collective advantage for advancing
  communications and intelligence},'' \emph{IEEE Wireless Communications},
  vol.~27, no.~4, pp. 96--102, 2020.

\bibitem{management}
R.~Li, W.~Liang, C.~Peng \emph{et~al.}, ``{Network AI management \&
  orchestration: A federated multi-task learning case},'' in \emph{2021 IEEE
  Globecom Workshops (GC Wkshps)}, Madrid, Spain, Dec. 2021.

\bibitem{fogIot}
M.~Chiang and T.~Zhang, ``Fog and iot: An overview of research opportunities,''
  \emph{IEEE Internet of things journal}, vol.~3, no.~6, pp. 854--864, 2016.

\bibitem{IIot}
W.~Z. Khan, M.~Rehman, H.~M. Zangoti \emph{et~al.}, ``Industrial internet of
  things: Recent advances, enabling technologies and open challenges,''
  \emph{Computers \& Electrical Engineering}, vol.~81, p. 106522, 2020.

\bibitem{fedavg}
B.~McMahan, E.~Moore, D.~Ramage \emph{et~al.}, ``{Communication-efficient
  learning of deep networks from decentralized data},'' in \emph{Artificial
  intelligence and statistics}, Fort Lauderdale, USA, Apr. 2017.

\bibitem{Keyboard}
A.~Hard, K.~Rao, R.~Mathews \emph{et~al.}, ``{Federated learning for mobile
  keyboard prediction},'' \emph{arXiv:1811.03604}, 2018.

\bibitem{machine}
Q.~Yang, Y.~Liu, T.~Chen \emph{et~al.}, ``{Federated machine learning: Concept
  and applications},'' \emph{ACM Transactions on Intelligent Systems and
  Technology (TIST)}, vol.~10, no.~2, pp. 1--19, 2019.

\bibitem{wireless}
N.~H. Tran, W.~Bao, A.~Zomaya \emph{et~al.}, ``Federated learning over wireless
  networks: Optimization model design and analysis,'' in \emph{IEEE INFOCOM
  2019-IEEE Conference on Computer Communications}, Paris, France, 2019.

\bibitem{mtl}
R.~Caruana, ``Multitask learning,'' \emph{Machine learning}, vol.~28, pp.
  41--75, 1997.

\bibitem{mocha}
V.~Smith, C.-K. Chiang, M.~Sanjabi \emph{et~al.}, ``Federated multi-task
  learning,'' in \emph{Advances in neural information processing systems}, Los
  Angeles, USA, Dec. 2017.

\bibitem{adaptive}
S.~Wang, T.~Tuor, T.~Salonidis \emph{et~al.}, ``Adaptive federated learning in
  resource constrained edge computing systems,'' \emph{IEEE Journal on Selected
  Areas in Communications}, vol.~37, no.~6, pp. 1205--1221, 2019.

\bibitem{energyefficient}
Z.~Yang, M.~Chen, W.~Saad \emph{et~al.}, ``Energy efficient federated learning
  over wireless communication networks,'' \emph{IEEE Transactions on Wireless
  Communications}, vol.~20, no.~3, pp. 1935--1949, 2021.

\bibitem{TimeOptimization}
M.~Chen, H.~V. Poor, W.~Saad \emph{et~al.}, ``Convergence time optimization for
  federated learning over wireless networks,'' \emph{IEEE Transactions on
  Wireless Communications}, vol.~20, no.~4, pp. 2457--2471, 2021.

\bibitem{convergence}
C.~T. Dinh, N.~H. Tran, M.~N.~H. Nguyen \emph{et~al.}, ``Federated learning
  over wireless networks: Convergence analysis and resource allocation,''
  \emph{IEEE/ACM Transactions on Networking}, vol.~29, no.~1, pp. 398--409,
  2021.

\bibitem{fedwireless}
H.~H. Yang, Z.~Liu, T.~Q.~S. Quek, and H.~V. Poor, ``Scheduling policies for
  federated learning in wireless networks,'' \emph{IEEE Transactions on
  Communications}, vol.~68, no.~1, pp. 317--333, 2020.

\bibitem{Differential}
M.~Seif, R.~Tandon, and M.~Li, ``Wireless federated learning with local
  differential privacy,'' in \emph{2020 IEEE International Symposium on
  Information Theory (ISIT)}, Los Angeles, USA, June 2020.

\bibitem{powercontrol}
D.~Liu and O.~Simeone, ``Privacy for free: Wireless federated learning via
  uncoded transmission with adaptive power control,'' \emph{IEEE Journal on
  Selected Areas in Communications}, vol.~39, no.~1, pp. 170--185, 2021.

\bibitem{LearningAndCommunications}
M.~Chen, Z.~Yang, W.~Saad \emph{et~al.}, ``A joint learning and communications
  framework for federated learning over wireless networks,'' \emph{IEEE
  Transactions on Wireless Communications}, vol.~20, no.~1, pp. 269--283, 2021.

\bibitem{sdca}
S.~Shalev-Shwartz and T.~Zhang, ``Stochastic dual coordinate ascent methods for
  regularized loss minimization.'' \emph{Journal of Machine Learning Research},
  vol.~14, no.~2, 2013.

\bibitem{cocoa}
M.~Jaggi, V.~Smith, M.~Tak{\'a}c \emph{et~al.}, ``{Communication-efficient
  distributed dual coordinate ascent},'' in \emph{Advances in neural
  information processing systems}, Montreal, Canada, Dec. 2014.

\bibitem{cocoa2}
V.~Smith, S.~Forte, M.~Chenxin \emph{et~al.}, ``{CoCoA: A general framework for
  communication-efficient distributed optimization},'' \emph{Journal of Machine
  Learning Research}, vol.~18, pp. 1--49, 2018.

\bibitem{clientedge}
L.~Liu, J.~Zhang, S.~Song \emph{et~al.}, ``Client-edge-cloud hierarchical
  federated learning,'' in \emph{ICC 2020-2020 IEEE International Conference on
  Communications (ICC)}, June 2020.

\bibitem{strategies}
J.~Kone{\v{c}}n{\`y}, H.~B. McMahan, F.~X. Yu \emph{et~al.}, ``Federated
  learning: Strategies for improving communication efficiency,''
  \emph{arXiv:1610.05492}, 2016.

\bibitem{efficient}
M.~Chen, N.~Shlezinger, H.~V. Poor \emph{et~al.}, ``Communication-efficient
  federated learning,'' \emph{Proceedings of the National Academy of Sciences},
  vol. 118, no.~17, p. e2024789118, 2021.

\bibitem{multitask}
T.~Evgeniou and M.~Pontil, ``Regularized multi--task learning,'' in
  \emph{Proceedings of the tenth ACM SIGKDD international conference on
  Knowledge discovery and data mining}, Seattle, USA, Aug. 2004.

\bibitem{multitask-feature}
A.~Argyriou, T.~Evgeniou, and M.~Pontil, ``Multi-task feature learning,'' in
  \emph{Advances in neural information processing systems}, Vancouver, Canada,
  Dec. 2006.

\bibitem{multitask2}
L.~Jacob, J.-p. Vert, and F.~Bach, ``Clustered multi-task learning: A convex
  formulation,'' \emph{Advances in neural information processing systems},
  vol.~21, 2008.

\bibitem{zhou2010exclusive}
Y.~Zhou, R.~Jin, and S.~C.-H. Hoi, ``Exclusive lasso for multi-task feature
  selection,'' in \emph{Proceedings of the thirteenth international conference
  on artificial intelligence and statistics}.\hskip 1em plus 0.5em minus
  0.4em\relax Sardinia, Italy: JMLR Workshop and Conference Proceedings, May
  2010, pp. 988--995.

\bibitem{activity}
D.~Anguita, A.~Ghio, L.~Oneto \emph{et~al.}, ``A public domain dataset for
  human activity recognition using smartphones,'' in \emph{Proceedings of the
  21th international European symposium on artificial neural networks,
  computational intelligence and machine learning}, Bruges, Belgium, Apr. 2013.

\end{thebibliography}
\end{document}